\documentclass[12pt]{article}

\usepackage{datetime}

\usepackage{epsfig,epic,eepic}
\usepackage{graphicx}
\usepackage{multicol}
\usepackage[all]{xy}
\usepackage{multirow}
\usepackage{rotating}
\usepackage{eurosym}
\usepackage{datetime}
\usepackage{hyperref}
\usepackage{import}
\usepackage{amsfonts}
\usepackage{amsmath, amsthm, amssymb}

\usepackage{color}
\usepackage{times}
\usepackage{eurosym}
\usepackage{longtable}
\usepackage[table]{xcolor}

\usepackage{algorithm}
\usepackage[noend]{algorithmic}
\usepackage{xspace}

\usepackage{afterpage, longtable, footnote, framed, latexsym} %%%%%%%%IGOR
\usepackage{color}
\usepackage{balance}
\usepackage{latexsym}
\usepackage{graphicx}
\usepackage{epsfig}
\usepackage[hang]{subfigure}
\usepackage{algorithm}
\usepackage[noend]{algorithmic}
\usepackage[T1]{fontenc}
\usepackage{epic}
\usepackage{eepic}
\usepackage{amsfonts}
\usepackage{amssymb}
\usepackage{stmaryrd}
\usepackage[all]{xy}
\usepackage{booktabs}
\usepackage{multirow}
\usepackage{rotating}
\usepackage{multicol}

\newtheorem{theorem}{Theorem}

\newtheorem{definition}{Definition}
\newtheorem{example}{Example}
\newtheorem{remark}{Remark}

\algsetup{linenodelimiter=.} 
\algsetup{indent=1em}

\definecolor{light-gray}{gray}{0.90}
\newcounter{theorem-backup}

\newcommand{\Z}{\mathbb{Z}}
\newcommand{\R}{\mathbb{R}}

\newcommand{\B}{\mathbb{B}}

\usepackage{listings}
\lstdefinelanguage{PseudoC}[ISO]{C++} { 
morekeywords={foreach, and, not, or, is, FIFO_Queue, HashTable, FILE, Cache}, 
morekeywords=[2]{HashInsert, Enqueue, Dequeue, next}, 
mathescape=true,
alsoletter={'},
deletestring=[b]'
}

\lstdefinelanguage{PseudoCWithNumbers}[ISO]{C++} { 
morekeywords={foreach, and, not, or, is, FIFO_Queue, HashTable, FILE, Cache}, 
morekeywords=[2]{HashInsert, Enqueue, Dequeue, next}, 
mathescape=true,
alsoletter={'},
deletestring=[b]'
}

\lstdefinelanguage{Murphi}[]{Pascal} { 
morekeywords={ruleset, rule, invariant, startstate, return, endif, endfor, endswitch, forall, endforall, exists, endexists}, 
mathescape=true, 
morestring=[b]",
morestring=[b]', 
morecomment=[s]{/*}{*/} ,
morecomment=[l]{--}
}

\lstdefinelanguage{PRISM}[ISO]{C++} { 
morekeywords={probabilistic, stochastic, const, rate, module, endmodule, init, P, U},
mathescape=true, 
alsoletter={'},
deletestring=[b]'
}

\lstdefinelanguage{Yacc}[ISO]{C++} { 
morekeywords={token, left}, 
mathescape=false,
alsoletter={'},
deletestring=[b]'
}

\lstdefinestyle{PseudoC}{
language=PseudoC,
basicstyle=\ttfamily,%\small, 
tabsize=1,
showlines=false,
emptylines=*1,
breaklines=true,
breakindent=5pt,
keywordstyle=\rmfamily\bfseries,%\normalsize, 
keywordstyle=[2]\rmfamily,%\normalsize, 
commentstyle=\itshape, 
columns=fixed,
showspaces=false, 
showstringspaces=false, 
showtabs=false, 
escapechar=\%
}

\lstdefinestyle{PseudoCWithNumbers}{
language=PseudoC,
basicstyle=\ttfamily,%\small, 
tabsize=1,
showlines=false,
emptylines=*1,
breaklines=true,
breakindent=5pt,
keywordstyle=\rmfamily\bfseries,%\normalsize, 
keywordstyle=[2]\rmfamily,%\normalsize, 
commentstyle=\itshape, 
columns=fixed,
showspaces=false, 
showstringspaces=false, 
showtabs=false, 
escapechar=\%,
numbersep=5pt,framexleftmargin=15pt,numbers=left,
}

\lstdefinestyle{Murphi}{
language=Murphi,
basicstyle=\ttfamily,%\small, 
tabsize=1,
showlines=false,
emptylines=*1,
breaklines=true,
breakindent=5pt,
keywordstyle=\rmfamily\bfseries,%\normalsize, 
commentstyle=\itshape, 
columns=fixed,
showspaces=false, 
showstringspaces=false, 
showtabs=false,
escapechar=\%}

\lstdefinestyle{PRISM}{
language=PRISM,
basicstyle=\ttfamily,%\small, 
tabsize=1,
showlines=false,
emptylines=*1,
breaklines=true,
breakindent=5pt,
keywordstyle=\rmfamily\bfseries,%\normalsize, 
commentstyle=\itshape, 
columns=fixed,
showspaces=false, 
showstringspaces=false, 
showtabs=false,
escapechar=\%}

\lstdefinestyle{Yacc}{
language=PseudoC,
basicstyle=\ttfamily,%\small, 
tabsize=1,
showlines=false,
emptylines=*1,
breaklines=true,
breakindent=5pt,
keywordstyle=\rmfamily\bfseries,%\normalsize, 
keywordstyle=[2]\rmfamily,%\normalsize, 
commentstyle=\itshape, 
columns=fixed,
showspaces=false, 
showstringspaces=false, 
showtabs=false, 
}

\newcommand{\qks}{\mbox{\sl QKS}}
\newcommand{\qkssc}{\mbox{\sl QKS}$^{\,\sf sc}$}

\newcommand{\fun}[1]{{\textsl{#1}}}

\definecolor{Blue}{rgb}{0.25,0.33,0.77}
\definecolor{Red}{rgb}{1,0,0}
\definecolor{Black}{rgb}{0,0,0}
\definecolor{Green}{rgb}{0,1,0}

\algsetup{linenodelimiter=.} 
\algsetup{indent=1em}

\definecolor{light-gray}{gray}{0.90}

\title{On Model Based Synthesis of Embedded Control Software}
\author{Vadim Alimguzhin, Federico Mari, Igor Melatti, Ivano Salvo, Enrico Tronci\\
\small \itshape Department of Computer Science\\
\small \itshape Sapienza University of Rome\\
\small \itshape via Salaria 113, 00198 Rome\\
\small email: \{alimghuzin,mari,melatti,salvo,tronci\}@di.uniroma1.it\\
\small Vadim Alimghuzin is also with the Department of Computer Science and Robotics\\
\small Ufa State Aviation Technical University\\
\small 12 Karl Marx Street, Ufa, 450000, Russian Federation}

\begin{document}

\maketitle

  \begin{abstract}
Many \emph{Embedded Systems} are indeed 
\emph{Software Based Control Systems} (SBCSs), that
is control systems whose controller consists of control software
running on a  microcontroller device.  This motivates investigation on
\emph{Formal Model Based Design} approaches for control software.
Given the formal model of a plant as a {\em Discrete Time Linear Hybrid System} and 
the implementation specifications 
(that is, number of bits in the \emph{Analog-to-Digital} (AD) conversion) 
correct-by-construction control software 
can be automatically generated from System Level Formal Specifications 
of the closed loop system (that is, \emph{safety} and \emph{liveness} requirements), 
by computing a suitable finite abstraction 
of the plant.

With respect to given implementation specifications, 
the automatically generated code implements a time optimal control strategy 
(in terms of set-up time),  
has a \emph{Worst Case Execution Time} 
linear in the number of AD bits $b$, but unfortunately, its size 
%of the automatically generated control software 
grows exponentially with respect to $b$. 
In many embedded systems, there are severe restrictions on the computational resources 
(such as memory or computational power) 
available to microcontroller devices.

This paper addresses model based synthesis of control software by trading
system level non-functional requirements (such us optimal set-up time, ripple) 
with software non-functional requirements (its footprint). 
Our experimental results 
%on the inverted pendulum and the multi-input buck DC-DC converter, 
show the effectiveness of our approach: for the inverted pendulum benchmark, 
by using a quantization schema with 12 bits, 
the size of the small controller is less than 6\% of the size of the time optimal one.
%\vspace{-1mm}
\end{abstract}

\clearpage
  %\tableofcontents
 
  \section{Introduction}

        Many \emph{Embedded Systems} are indeed \emph{Software Based Control Systems} (SBCSs).
An SBCS %is a {\em closed loop system} that 
consists of two main subsystems, 
the \emph{controller} and the \emph{plant}, 
that form a {\em closed loop system}.
Typically, the plant is a physical system
consisting, for example, of mechanical or electrical devices
whereas the controller 
consists of \emph{control software} running on a microcontroller.
%In an endless loop, 
%the controller
%reads \emph{sensor} outputs from the plant and
%sends commands to plant \emph{actuators}
%in order to guarantee that the 
%\emph{closed loop system}
%(that is, the system consisting of both
%plant and controller) meets given
%\emph{safety} and \emph{liveness} specifications
%(\emph{System Level Formal Specifications}).
%
Software generation from models
and formal specifications forms the core of
\emph{Model Based Design} of embedded software %\marginpar{cancellare ``embedded''? - NO et}
\cite{Henzinger-Sifakis-fm06}.
This approach is particularly interesting for SBCSs
since in such a case {\em System Level Formal Specifications} are much easier
to define than the control software behavior itself. 
% control system def
%Embedded systems are often safety or mission critical control systems.
%A \emph{control system} %\marginpar{cancellato ``embedded'' - OK et}
%(e.g. see \cite{modern-control-theory-1990}) 
%consists of two subsystems 
%(forming the \emph{closed loop} system):
%the \emph{controller} and the \emph{plant}. 
%In an endless 
%loop the controller measures outputs from and sends commands to the
%plant in order to drive it towards a given 
%\emph{goal}.
%
% plant is hardware, controller is software
%The plant usually is a physical system,  
%for example it can be a mechanical system or an electrical circuit. 
%whereas 
%A \emph{Software Based Control System} (SBCS) is a control system
%in which the controller consists of \emph{control software} running
%on a microcontreller. Many embedded systems as well as many
%control systems are indeed SBCSs.
The typical control loop skeleton for an SBCS is the following.
Measure $x$ of the system state from plant \emph{sensors} go through an
\emph{analog-to-digital} (AD) conversion, yielding a {\em quantized} value $\hat{x}$. 
A function \fun{ctrlRegion} checks if $\hat{x}$ belongs to the region in which 
the control software works correctly. 
If this is not the case a \emph{Fault Detection, Isolation and 
Recovery} (FDIR) procedure is triggered, otherwise a function \fun{ctrlLaw} computes a 
command $\hat{u}$ to be sent to plant \emph{actuators} after a \emph{digital-to-analog} (DA) 
conversion.
Basically, the control software design problem for SBCSs consists
in designing software implementing functions
%\marginpar{meglio ``functions'' che ``procedures''}
\fun{ctrlLaw} and
\fun{ctrlRegion} in such a way that the 
closed loop system 
%(that is, the system consisting of both
%plant and controller) 
meets given
\emph{safety} and \emph{liveness} specifications.

%\input{control-loop-figure.tex}

%\subsection{The Separation-of-Concerns Approach}
%\label{separation-of-concerns}

For SBCSs, system level specifications are typically
given with respect to the desired behavior 
of the closed loop system. The control software 
%\emph{control software} 
%(that is, \fun{ctrlLaw} and \fun{ctrlRegion})
is designed using a 
\emph{separation-of-concerns} approach. That is, 
\emph{Control Engineering} techniques 
(e.g., see \cite{modern-control-theory-1990})
are used
%IM_0413<b>
%to design \emph{control laws} (i.e. a functional
to design, from the closed loop system level specifications,
\emph{functional specifications} (\emph{control law})
for the control software whereas 
\emph{Software Engineering} techniques are used to
design control software implementing the given
functional specifications.
Such a separation-of-concerns approach has 
several
%the following
drawbacks. 

First, usually control engineering techniques do not
yield a formally verified specification for the
control law %or controllable region 
when quantization is taken into account. 
This is particularly the case when the plant has to be modelled 
as a \emph{Hybrid System}, 
that is a system with continuous as well as discrete state changes
\cite{alg-hs-tcs95,HHT96,AHH96}.
As a result, even if the control software meets its functional
specifications there is no formal guarantee that
system level specifications are met since
quantization effects are not formally accounted for.

Second, issues concerning computational resources, 
such as control software \emph{Worst Case Execution Time} (WCET),
can only be considered very late in the SBCS design activity, namely once the 
software has been designed. As a result, the control software may have a WCET
greater than the sampling time. 
%(line \ref{sampling-and-hold} in Fig. \ref{control-loop-figure.tex}). 
This invalidates the schedulability
analysis (typically carried out before the control software is completed)
and may trigger redesign of the software or even 
of its functional specifications
(in order to simplify its design).

Last, but not least, the classical 
separation-of-concerns approach does not effectively support 
design space exploration for the control software. 
In fact, although in general there will be many functional specifications
for the control software that will allow meeting the given
system level specifications, the software engineer only
gets one 
%functional specification
 to play with. This overconstrains a priori 
 the design space
for the control software implementation preventing, for example, effective performance trading 
(e.g., 
between number of bits in AD conversion,  
WCET, 
RAM usage, 
CPU power consumption, 
etc.).
% OBSERVATION ABOUT SIMULINK - TAGLIATA FINAL VERSION (IVANO)
%We note that the above considerations also apply
%to the typical situation where
%Control Engineering techniques are used to design a control
%law and then tools like Simulink are used to generate the control
%software. 

%TOLTO: IVANO
%Even when the control law is automatically generated
%and proved correct
%(for example, as in \cite{pessoa-cav10}) such an approach
%does not yield any formal guarantee about the software correctness
%since quantization of the state measurements
%is not taken into account in the computation of the control
%law. 
%Thus such an approach cannot answer questions like:
%1) Can 8 bit AD be used or instead we need, say, 12 bit AD?
%2) Will the control software code run \emph{fast enough} 
%   on a, say, 1 MIPS microcontroller 
%(that is, is the control software WCET less than the sampling time)?
%3) What is the controllable region?

        %\vspace{-1mm}
\subsection{Motivations}

The previous considerations motivate research on 
Software Engineering methods and tools focusing on control software synthesis
rather than on control law %synthesis 
as in Control Engineering. 
The objective is that from the plant model (as a hybrid system), 
from formal specifications for the closed loop system behavior 
%(\emph{System Level Formal Specifications}) 
and from \emph{Implementation Specifications} 
(that is, the number of bits used in the quantization process)
such methods and tools can generate correct-by-construction control software
satisfying the given specifications.

A {\em Discrete Time Linear Hybrid System} (DTLHS) is a discrete time hybrid system
whose dynamics is modeled as a {\em linear predicate} 
over a set of continuous as well as discrete 
variables that describe system state, system inputs and disturbances.
System level safety as well as liveness specifications are modeled as sets of states
defined, in turn, as predicates. 
%In our setting, as always in control problems, liveness constraints define the
%set of states that any evolution of the closed loop system should eventually reach
%(\emph{goal states}).
By adapting the proofs in \cite{decidability-hybrid-automata-jcss98} 
 for the reachability problem in dense time hybrid systems, 
 it has been shown that the control synthesis problem 
 is undecidable for DTLHSs~\cite{ictac2012}.  
%as can be shown by adapting the proofs in \cite{decidability-hybrid-automata-jcss98} 
% for the reachability problem in dense time Hybrid Systems. 
Despite that,  
 non complete or semi-algorithms usually succeed in finding controllers for 
 meaningful hybrid systems.
 
The tool \qks\  \cite{qks-cav10} automatically synthesises control software  
starting from a plant model given as a DTLHS,  %s{\em Discrete Time Linear Hybrid System} (DTLHS), 
the number of bits for AD conversion, and   
System Level Formal Specifications of the closed loop system. 
The generated code, however, may be very large, since it grows exponentially with 
the number of bits of the quantization schema \cite{icsea2011}. 
On the other hand, controllers synthesised by considering a 
finer quantization schema usually have a better behaviour with respect to 
many other non-functional requirements, such as {\em ripple} and 
{\em set-up time}.
Typically, a microcontroller device in an Embedded System  
has limited resources in terms of computational power and/or memory. 
Current state-of-the-art microcontrollers have up to 512Kb of memory, and other 
design constraints (mainly costs) may impose to use even less powerful devices.
As we will see in Sect.~\ref{sec:expres}, 
by considering a quantization schema with 12 bits on the inverted pendulum system, 
\qks\  generates a controller which has a size greater than 8Mbytes.

This paper addresses model based synthesis of control software by trading
system level non-functional requirements 
with software non-functional requirements. Namely,   
we aim at reducing the code footprint, possibly at the cost of having a 
suboptimal set-up time and ripple. 
%In this paper, we focus our attention on reducing the footprint of 
%control software code, by trading with 
%system level non-functional requirements (such us optimal set-up time, ripple). 
%%for Embedded Systems.    

        %\marginpar{Pero' il problema e' indecidibile, non e' vero che se c'e' noi lo troviamo, come questa frase parrebbe implicare;
%usare (semi)algorithm? Comunque, dire che il problema e' indecidibile}
%\vspace*{-1mm}
\subsection{Our Main Contributions}
Fig.~\ref{fig:cssf} shows the model based control software synthesis flow that we consider in 
this paper. 
A specification consists of a plant model, given as a DTLHS, 
System Level Formal Specifications that describe functional requirements
of the closed loop system, and Implementation Specifications that describe 
non functional requirements of the control software, such as the number of bits 
used in the quantization process, the required WCET, etc.
%Fig.~\ref{fig:cssf} summarizes the main steps that 
In order to generate the control software, the tool \qks\ takes the following steps. 
%starting from specifications.  
First (step 1), a suitable finite discrete abstraction 
({\em control abstraction}~\cite{qks-cav10}) $\hat{\cal H}$ of the DTLHS plant model ${\cal H}$ is computed;  
$\hat{\cal H}$ depends on the quantization schema and it is the plant as it can be seen from the control software after 
AD conversion.
Then (step 2), given an abstraction $\hat{G}$ of the goal states $G$, it computes  
a controller $\hat{K}$ that starting from any initial abstract state,
 drives $\hat{\cal H}$ to $\hat{G}$ regardless of possible nondeterminism.
Control abstraction properties ensure that $\hat{K}$ %a controller for $\hat{\cal H}$ 
is indeed a (quantized representation of a) controller for the original plant ${\cal H}$.
Finally (step 3), the finite automaton $\hat{K}$ is translated into control software (C code).  
Besides meeting functional specifications, the generated control software meets some 
non functional requirements: it implements a (near) time-optimal control strategy, 
and it has a WCET guaranteed to be linear in the number of bits of the 
quantization schema.

\begin{figure}
  \centering
  \hspace{-.6cm}\includegraphics[width=0.8\columnwidth]{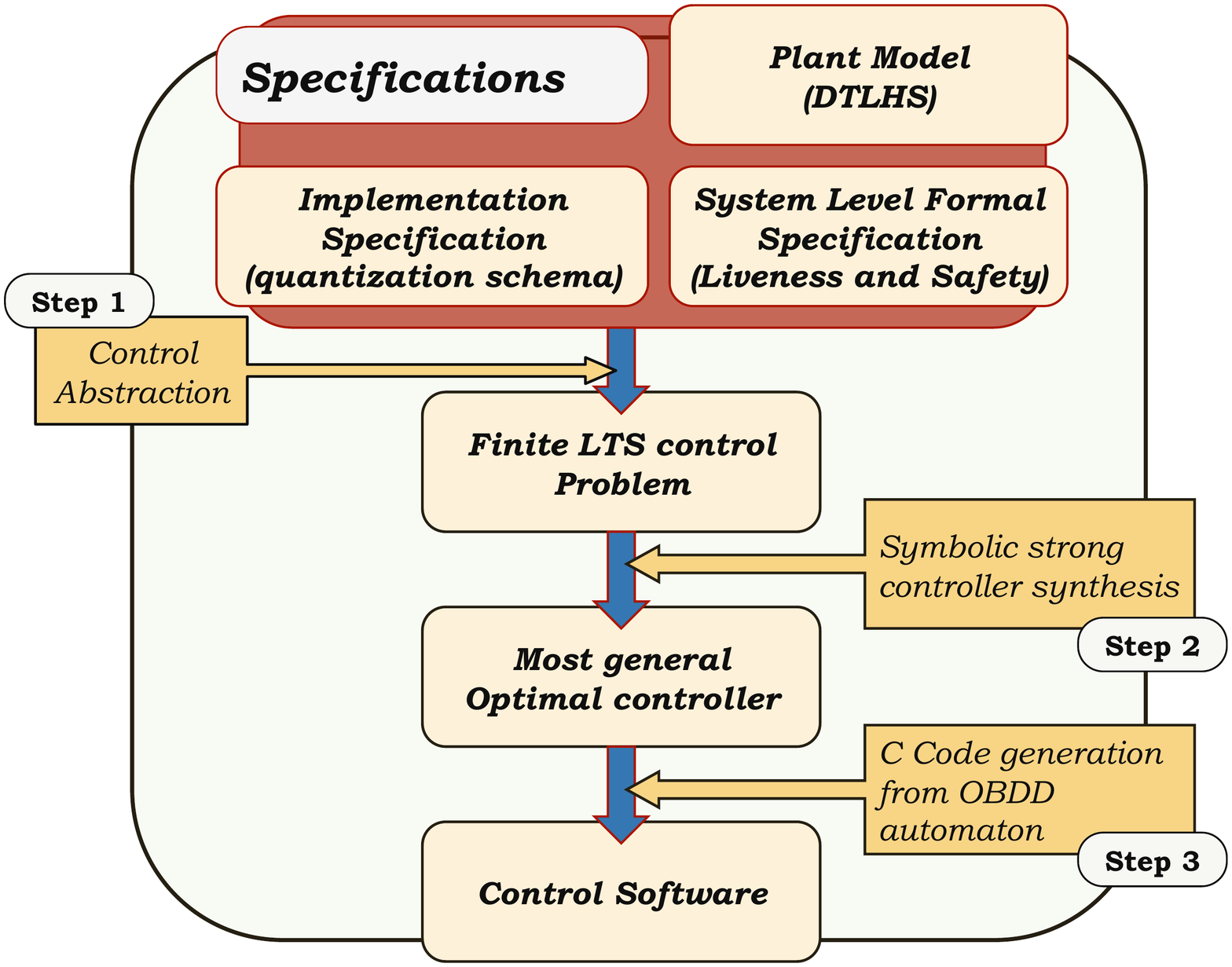}
  %\vspace*{-0.8cm}
  \caption{Control Software Synthesis Flow.}
  \label{fig:cssf}
  % \caption{Inverted Pendulum with Stationary Pivot Point and Massless Rod}
%\vspace*{-0.5cm}
\end{figure}

To find the quantized controller $\hat{K}$, \qks\ implements 
the symbolic synthesis algorithm in \cite{strong-planning-98}, 
based on {\em Ordered Bynary Decision Diagrams} (OBDDs) manipulation.  
This algorithm finds %in \cite{strong-planning-98} finds
a time-optimal solution, %to a given finite state control problem, 
i.e. the controller 
$\hat{K}$ drives the system $\hat{\cal H}$ to $\hat{G}$ always along shortest paths. 
The finer the control abstraction is (i.e. when the quantization schema is more precise), 
the better is the control strategy found. Unfortunately, such time optimal control strategies 
may lead to very large controllers in terms of the size of the generated C control software.

Driven by the intuition that by enabling the very same action on 
large regions of the state space we may decrease the control software size,  
we design a controller synthesis algorithm 
(Alg.~\ref{alg:smallctr4} in Sect.~\ref{sec:control-synthesis-algo})
that gives up optimality and looks for maximal regions that can be controlled by performing the same action. 
We formally prove its correctness %(Theor.~\ref{thm:correctness}) 
and completeness (Theor.~\ref{thm:correctness} and~\ref{thm:maximal} in Sect.~\ref{sec:algo-correctness}).

Experimental results in Sect.~\ref{sec:expres} 
show that such a heuristic effectively mitigates the exponential growth
of the controller size without having a significant impact on non-functional system level 
requirements such as set-up time and ripple. 
We accomplish this result without changing the WCET of the synthesized control software.  
%, while the WCET is still linear on the number of 
%bit of the quantization schema.   
%such as WCET and ripple. %with respect to the number of bits used in the quantization schema. 
For the inverted pendulum benchmark, by using a quantization schema with 12 bits, 
the size of our controller is less than $6\%$ of the size of the time optimal controller.

%The price to pay is that the control strategy is no more optimal. However we will show 
%that the average set-up time of $\blacktriangleright$.

%\vspace*{-1mm}
\subsection{Related Work}
\label{related-works.tex}
Control Engineering has been studying
control law design (e.g., optimal control, robust control, etc.),
for more than half a century
(e.g., see~\cite{modern-control-theory-1990}).
Also \emph{Quantized Feedback Control} 
has been widely studied in control engineering
(e.g. see \cite{quantized-ctr-tac05}).
However such research does not address hybrid systems 
%(our case) 
and, as explained above, focuses on control law design rather
than on control software synthesis. %(our goal).
%As explained in Sect.~\ref{separation-of-concerns} such results
%cannot be directly used in our (formal) software synthesis context.
%On the other hand we note that there are many control systems that 
%are not software based (e.g., in analog circuit design).
%In such cases, of course, our approach cannot be used.
Traditionally, control engineering approaches
model \emph{quantization errors} as statistical \emph{noise}. 
As a result, correctness of the control law holds in a probabilistic sense.
Here instead, we model quantization errors as nondeterministic 
(\emph{malicious})
\emph{disturbances}. This guarantees system level correctness of the generated
control software (not just that of the control law) 
with respect to \emph{any} possible sequence of quantization errors.
%Indeed, we synthesize control software which is \emph{robust} also
%with respect to nondeterministic variations in the plant parameters.
%Indeed, to the best of our knowledge, 
%no previously published work on hybrid system control
%addresses the issue of automatic generation
%of corret-by-construction control software satisfying given 
%\emph{implementation specifications} 
%(i.e., number of bits used in state feedback quantization) and 
%\emph{System Level Formal Specifications}.

Formal verification of 
\emph{Linear Hybrid Automata} (LHA)
% \cite{alg-hs-tcs95,AHH96}
% reachability and existence of a control law are both undecidable problems
% \cite{dt-ctr-rect-aut-icalp97,decidability-hybrid-automata-jcss98}.
\cite{alg-hs-tcs95} %,AHH96,dt-ctr-rect-aut-icalp97,decidability-hybrid-automata-jcss98} 
%This, of course, has not prevented devising 
has been investigated in
%\cite{AHH96,HHT96,phaver-sttt08,ctr-lha-cdc97,TLS99,Benerecetti:2011yq}.
\cite{HHT96,phaver-sttt08,ctr-lha-cdc97,TLS99}. %,Benerecetti:2011yq}.
Quantization can be seen as a sort of abstraction.
In a hybrid systems formal verification context, abstractions 
has been widely studied 
%in a hybrid system 
(e.g., see \cite{AHLP-ieee00,pred-abs-tecs06}), 
%JKC-hscc07,JhaBS07,cegar-hs-cav07,cegar-par-synth-hscc08,alur-hs-cegar-tcs06,FKR-date06,mari-tronci-hscc07,TazImu09}).
%Note however that 
%In a verification context abstractions are designed so as 
to ease the verification task. On the other hand, 
in control software synthesis,  
quantization is a design requirement since it models a hardware component
(AD converter) which is part of the specification of the control software synthesis problem.
%Indeed, in our setting, we have to design a controller \emph{notwithstanding} the nondeterminism stemming
%from the quantization process. 
As a result, %the techniques used to devise 
clever abstractions considered 
in a verification setting cannot be directly used in our synthesis setting where quantization is given.

The abstraction--based approach to controller synthesis has also been broadly 
investigated.
Based on a notion of suitable finite state abstraction 
(e.g. see \cite{pola-girard-tabuada-2008}) %tabuada-quantized-hscc07,
control software synthesis for continuous time linear systems
(no switching) has been implemented in the tool 
{\sc Pessoa} \cite{pessoa-cav10}. 
On the same wavelength, \cite{Yordanov:2012ul} generates a
control strategy from a finite abstraction of a  
\emph{Piecewise Affine Discrete Time Hybrid System} (PWA-DTHS). 
Also the Hybrid Toolbox \cite{HybTBX} considers PWA-DTHSs. 
Such tools output a feedback control law that is 
then passed to Matlab in order to generate control software.
Finite horizon control of PWA-DTHSs has been studied using
a MILP based approach (e.g. see \cite{sat-opt-ctr-hscc04}).
Explicit finite horizon control synthesis algorithms for discrete time 
(possibly non-linear) hybrid systems have been studied in 
\cite{icar08}. %and citations thereof.
All such approaches do not account for state feedback 
quantization since they all assume \emph{exact} (i.e. real valued) state measures.
% and 
%thus, %as explained above, 
%they do not
%offer any formal guarantee about system level correctness 
%of the generated software. %, which is instead our focus here.
Optimal switching logic, i.e. 
synthesis of optimal controllers with respect to some cost function 
has also been widely investigated 
(e.g. see \cite{jha-emsoft11}). %and citations thereof. 
In this paper, we focus on non-functional sofware requirements 
rather than non-functional system-level requirements.

Summing up, to the best of our knowledge, no previously published
result is available about model based synthesis of small footprint control software 
from a plant model, system level specifications and implementation specifications.

\section{Control Software Synthesis} 
To make this paper self-contained, 
first we briefly summarize previous work on automatic generation of 
control software for {\em Discrete Time Linear Hybrid Systems} (DTLHSs) 
from System Level Formal Specifications. We focus on 
basic definitions and mathematical tools 
that will be useful later. 

%Fig.~\ref{fig:cssf} shows the control software synthesis flow that we consider here.
We model the controlled system (i.e. the plant) as a DTLHS (Sect.~\ref{sec:dtlhs}), 
that is a discrete time hybrid system 
whose dynamics is modeled as a {\em linear predicate} (Sect.~\ref{subsection:predicates})  
over a set of continuous as well as discrete 
variables.
% that describe system state, system inputs and disturbances. 
The semantics of a DTLHS is given in terms of a {\em Labeled Transition Systems} 
(LTSs, Sect.~\ref{sec:lts}).

Given a plant ${\cal H}$ modeled as a DTLHS, a set of {\em goal states} $G$ 
({\em liveness specifications}) and an {\em initial region} $I$, 
both represented as linear predicates, 
we are interested in finding a {\em restriction $K$ of the behaviour} 
of ${\cal H}$ such that in the {\em closed loop system} %${\cal H}^{(K)}$, 
all paths starting in $I$ lead to $G$ after a finite number of steps. 
Moreover, we are interested in controllers that take their decisions by looking at 
{\em quantized states}, i.e. the values that the control software 
reads after an AD conversion. This is the {\em quantized control problem} 
(Sect.~\ref{sec:qcp}).

\sloppy 

The quantized controller is computed by solving an {\em LTS control problem} 
(Sect.~\ref{sec:lts-ctr}), by using a symbolic approach based on {\em Ordered Binary Decision Diagrams} (OBDDs) (Sect.~\ref{sec:ctr-syn}). Finally, we briefly describe how C control software is automatically generated from the OBDD controller representation (Sect.~\ref{sec:ctr-sftw-gen}).   

\fussy

%Finding $K$ is the DTLHS {\em control problem} (Sect.~\ref{sec:dtlhs-ctr}) that is in 
%turn defined as a 
%The DTLHS {\em control problem} is defined in terms of a 
%suitable LTS control problem (Sect.~\ref{sec:lts-ctr}).

%Finally,  

  		\subsection{Predicates}
\label{subsection:predicates}
We denote with $[n]$ an initial segment $\{1,\ldots, n\}$ of the natural numbers. 
We denote with $X$ = $[x_1, \ldots, x_n]$ a
finite sequence of distinct variables, that we may regard, when convenient, as a set.
%By abuse of language we may regard sequences as sets and
%we use $\cup$ to denote list concatenation.
%Unless otherwise stated 
Each variable $x$ ranges on a known (bounded or unbounded)
interval ${\cal D}_x$ either of the reals or of the integers (discrete
variables). Boolean variables are discrete variables
ranging on the set $\B$ = \{0, 1\}. 
%We denote with $\sup(x)$ ($\inf(x)$) the $\sup$ ($\inf$)
%of ${\cal D}_x$.
We denote with ${\cal D}_X$ the set $\prod_{x\in X} {\cal D}_x$.
To clarify that a variable $x$ is {\em continuous} 
(resp. discrete, boolean)
%(i.e. real valued) 
we may write $x^{r}$ (resp. $x^{d}$, $x^{b}$). 
%Similarly, 
%to clarify that a variable $x$ is {\em discrete} 
%(i.e. integer valued) we may write $x^{d}$. 
%We may write $x^{b}$ to denote a boolean variable.
Analogously $X^{r}$ ($X^{d}$, $X^{b}$) denotes the sequence
of real (integer, boolean) variables in $X$.
Unless otherwise stated, we suppose
${\cal D}_{X^r} = \R^{|X^r|}$ and ${\cal D}_{X^d} = \Z^{|X^d|}$. 
Finally, if $x$ is a boolean variable 
we write $\bar{x}$ for $(1 - x)$.

%\vspace*{-4mm}
%\subsection{Predicates}
%\label{subsection:predicates}
%\vspace*{-2mm}
 
A {\em linear expression} $L(X)$ over a list of variables $X$ is a linear
combination of variables in $X$ with rational coefficients. 
%
% For example, let $X$ = [$x$, $y$, $z$]. Then $3x + 4.1y - 0.3z$ is a
% linear expression.
A {\em linear constraint} over $X$ (or simply a {\em constraint})  is an expression of the form
$L(X) \leq b$,
%$L(X) \bowtie b$ 
where %$L(X)$ is a linear expression over $X$ 
%$\bowtie$ is one of $\leq$, $\geq$, $=$  
%and 
$b$ is a rational constant. In the following, we also write $L(X) \geq b$
for $-L(X) \leq -b$, $L(X) = b$ for ($L(X) \leq b$) $\wedge$ ($L(X) \geq b$), 
and $a \leq x \leq b$ for $x \geq a \;\land\; x \leq
b$.
%  A {\em constraint} (over a list of variables $X$) is an expression of
% the form $\alpha \leq b$ where $\alpha$ is a linear expression and $b$
% is a real constant.
% For example, with $X$ as above, $3x + 4.1y - 0.3z \leq 3$ is a
% constraint.  We now give the recursive definition of

{\em Predicates} are inductively defined as follows.
A constraint $C(X)$ over a list of variables $X$ is a predicate over 
$X$. 
If $A(X)$ and $B(X)$ are predicates over $X$, then $(A(X) \land B(X))$
and $(A(X) \lor B(X))$ are predicates over X.  Parentheses may be
omitted, assuming usual associativity and precedence rules of logical
operators.
A {\em conjunctive predicate} is a conjunction of constraints.

% A {\em valuation} $X^\ast\in {\cal D}_X$ over a list of variables $X$\marginpar{punto di ${\cal D}_X$ o funzione???}
% is a function $v$ that maps each variable $x \in X$ to a value $v(x)$
% in ${\cal D}_x$.
% A \emph{satisfying assignment} to a predicate $P$ over $X$ is a
% valuation $X^{*}$ such that $P(X^{*})$ holds. 
% Abusing notation, we may
% denote with $P$ the set of satisfying assignments to the predicate 
% $P(X)$. Two predicates $P$ and $Q$ over $X$ are {\em equivalent},
% notation $P\equiv Q$ if they have the same set of 
% satisfying assignments.

A {\em valuation} over a list of variables $X$
is a function $v$ that maps each variable $x \in X$ to a value $v(x)
\in {\cal D}_x$.
Given a valuation $v$, we denote with $X^\ast\in {\cal D}_X$ the sequence of values 
$[v(x_1),\ldots,v(x_n)]$. %By abuse of language, 
We also call valuation the sequence of values $X^\ast$.
A \emph{satisfying assignment} to a predicate $P(X)$ is a
valuation $X^{*}$ such that $P(X^{*})$ holds. If a satisfying assignment to a
predicate $P$ over $X$ exists, we say that $P$ is {\em feasible}.
Abusing notation, we may
denote with $P$ the set of satisfying assignments to the predicate 
$P(X)$. 

Two predicates $P$ and $Q$ over $X$ are {\em equivalent},
denoted by $P\equiv Q$, if they have the same set of 
satisfying assignments. Two predicates $P(X)$ and $Q(Z)$ are 
{\em equisatisfiable}, notation $P\simeq Q$ if 
$P$ is satisfiable if and only if $Q$ is satisfiable.
A variable $x\in X$ 
is said to be {\em bounded} in $P$ if 
%$\sup\{x \; | \; P(X)\}$ and $\inf\{x \; | \; P(X)\}$
%are both finite.
there exist $a$, $b \in {\cal D}_x$
such that $P(X)$ implies $a \leq x \leq b$.
A predicate is bounded if all its variables are bounded.
%
%
%
%
%Let $a$ be a real number and $x$ be a bounded variable. 
%We write 
%$\sup_{P(X)}(a x)$ [$\inf_{P(X)}(a x)$]
%for 
%$a\sup\{x \; | \; P(X)\}$ [$a\inf\{x \; | \; P(X)\}$]
%when $a \geq 0$ and for 
%$a\inf\{x \; | \; P(X)\}$ [$a\sup\{x \; | \; P(X)\}$]
%when $a < 0$. 
%
%Analogously, we write $\inf(a x)$ for $a \inf(x)$ if $a \geq 0$ and
%for $a \sup(x)$ if $a < 0$.  
%Let $L(X)$ = $\sum_{i=1}^{n}a_ix_i$ be 
%a linear expression and $P(X)$ be a bounded predicate. 
%We write $\sup_{P(X)}(L(X))$ for
%$\sum_{i=1}^{n} \sup_{P(X)}(a_i x_i)$ and 
%$\inf_{P(X)}(L(X))$ for $\sum_{i=1}^{n} \inf_{P(X)}(a_i x_i)$.

Given a constraint $C(X)$ and a fresh boolean variable ({\em guard}) $y \not\in X$,
the {\em guarded constraint} $y \to C(X)$ (if $y$ then $C(X)$) denotes
the predicate $(y = 0) \lor C(X)$. Similarly, we use $\bar{y} \to
C(X)$ (if not $y$ then $C(X)$) to denote the predicate $(y = 1) \lor
C(X)$.
A {\em guarded predicate} is a conjunction of 
either constraints or guarded constraints.
It is possible to show that, 
if a guarded predicate $P$ is bounded,
then $P$ can be transformed into an equisatisfiable conjunctive predicate. % \rimandotekrep{proof-pred.tex}. 

\subsection{Labeled Transition Systems}
\label{lts.tex}
\label{sec:lts}
%\vspace*{-1mm}

A \emph{Labeled Transition System} (LTS) is a tuple
${\cal S} = (S, A, T)$ where 
$S$ is a (possibly infinite) set of states, 
$A$ is a (possibly infinite) set of \emph{actions}, and 
$T:S\!\times\! A\!\times\! S \!\rightarrow\!\B$
is the \emph{transition relation} of ${\cal S}$.
%We say that $T$ (and ${\cal S}$) is {\em deterministic} if $T(s, a, s') \land
%T(s, a, s'')$ implies $s' = s''$, and {\em nondeterministic}
%otherwise. 
%\footnote{Note that, with this definition, deterministic LTSs are not
%a special case of nondeterministic ones, as it usually is.}.
%is a relation such that $\forall \; s, a$ $\exists s' \;$ $T(s, a, s')$.
%\end{definition}
%Note that in our definition of LTS we ask that in any state
%all actions are possible. 
%Illegal actions in a state can be modeled as usual with transitions to a
%\emph{sink} state.
%In the following
%${\cal S} = (S, A, T)$ is an LTS,
Let $s \in S$ and $a \in A$.
%
%
%If $\cal S$ = ($S$, $A$, $T$) is an LTS we also write
%$\mbox{\rm State}({\cal S})$ for $S$, 
%$\mbox{\rm Act}({\cal S})$ for $A$ and
%$\mbox{\rm Trans}({\cal S})$ for $T$.
%
%
%
%Action $u \in \mbox{\rm Act}({\cal P})$ is \emph{admissible} in 
%state $s \in \mbox{\rm State}({\cal P})$ if there exists 
%$s' \in \mbox{\rm State}({\cal P})$ s.t. 
%$\mbox{\rm Trans}({\cal P})(s, u, s')$.
%
We denote with 
$\mbox{\rm Adm}({\cal S}, s)$ the set of actions
admissible in $s$, that is $\mbox{\rm Adm}({\cal S}, s)$ = $\{a \in A
\; | \; \exists s': T(s, a, s') \}$
and with
$\mbox{\rm Img}({\cal S}, s, a)$ the set of next
states from $s$ via $a$, that is $\mbox{\rm Img}({\cal S}, s, a)$ =
$\{s' \in S \; | \; T(s, a, s') \}$.
%We call {\em transition} a triple $(s, a, s') \in S$ $\times$ $A$ $\times$ $S$,
%and {\em self loop} a transition $(s, a, s)$. A transition $(s, a, s')$ [self
%loop $(s, a, s)$] is a
%{\em transition [self loop] of ${\cal S}$} iff $T(s, a, s')$ [$T(s, a, s)$].
%We call {\em self-loop} a transition of the form $T(s, a, s)$.
%%Let $a \in \mbox{\rm Act}({\cal S})$.
%%We denote with $\mbox{\rm Img}({\cal S}, s, a)$ the set 
%%of next states from $s$ via $a$, that is: 
%%$\mbox{\rm Img}({\cal S}, s, a)$ =
%%$\{s' \in \mbox{\rm State}({\cal S})\; | 
%%\; \mbox{\rm Trans}({\cal S})(s, a, s') \}$.
%%
%%
A {\em run} or \emph{path}
for an LTS ${\cal S}$ 
%= ($S$, $A$, $T$)
is a sequence 
$\pi$ =
%$s(0)$$a(0)$$s(1)$$a(1)$ $s(2)$$a(2)\ldots$ 
$s_0, a_0, s_1, a_1, s_2, a_2, \ldots$ 
%of states $s(t)$ and actions $a(t)$ 
of states $s_t$ and actions $a_t$ 
%of ${\cal S}$ 
such that
%$x(0) \in I$ and
%$\forall t \geq 0$ $T(s(t), a(t), s(t+1))$.
$\forall t \geq 0$ $T(s_t, a_t, s_{t+1})$.
The length $|\pi|$ of a finite run $\pi$ is the number of actions
%IM<b>
%is this used???????
%YES!!!!! <IS>
%IM<e>
in $\pi$. 
%The length of an infinite run is infinite.
We denote with $\pi^{(S)}(t)$ the $(t + 1)$-th state element of
$\pi$, and with $\pi^{(A)}(t)$ the $(t + 1)$-th action element of
%$\pi$. That is $\pi^{(S)}(t)$ = $s(t)$, and $\pi^{(A)}(t)$ = $a(t)$.
$\pi$. That is $\pi^{(S)}(t)$ = $s_t$, and $\pi^{(A)}(t)$ = $a_t$.
\subsubsection{LTS Control Problem}
\label{sec:lts-ctr}
A \emph{controller} for an LTS ${\cal S}$ 
%(Def. \ref{def:ctroller-lts}) 
is used to restrict the dynamics of ${\cal S}$ 
so that all states in the initial region %($I$) 
will eventually reach %in one or more steps 
the goal region. %($G$).
%In the following, 
We formalize such a concept
by defining the LTS control problem %strong %and weak 
and its solutions. 
In what follows, let ${\cal S} = (S, A, T)$ be an LTS, 
% ${\cal H} = (X, U, Y, N)$ be a DTLHS,
$I$, $G$ $\subseteq$ $S$ 
%[${\cal D}_{X}$] 
be, respectively, the {\em initial} and {\em goal} regions. %of ${\cal S}$.

\begin{definition}%[Controller for LTS]
%\vspace{-1mm}
  \label{def:ctroller-lts}
  \label{def:ctrproblem-lts}
  A \emph{controller} for 
%LTS 
${\cal S}$ is a function 
  $K : S \times A \rightarrow \B$
  such that $\forall s \in S$, $ \forall a \in A$, if $K(s, a)$ then
  $\exists s' \; T(s, a, s')$. 
  If $K(s,a)$ holds, we say that the action $a$ is {\em enabled} by $K$ in $s$. %the state $s$. 

  The set of states for which at least one  
  action is enabled is denoted by ${\rm dom}(K)$. Formally, ${\rm dom}(K)$ $=$ $\{s \in
  S \; | \; $$\exists a \; K(s, a)\}.$
  
  We call a controller $K$ a {\em control law} if $K$ enables at most one action in 
  each state.  
  Formally, $K$ is a control law if, for all $s\in{\rm dom}(K)$, 
  $K(s, a)$ and $K(s,b)$ implies $a=b$.
  % 
%  We denote with 

The \emph{closed loop system} 
  is the LTS ${\cal S}^{(K)}=(S, A, T^{(K)})$, where 
$T^{(K)}(s, a, s')$ $=$ $T(s, a, s') \wedge K(s, a)$.
%\end{definition}
% 
%\begin{itemize}
% not needed since all actions always possible.
%\item
%For all $s \in S$,  
%$a \in A$,
%if $K(s, a)$ then $\exists s' \; T(s, a, s')$.
%
%\item
%For all $s \in S$,  
%if $\exists a,\; s' \; T(s, a, s')$ 
%then
%$\exists a\;  K(s, a)$
%\end{itemize}
%
% We denote with $\mbox{\rm Dom}(K)$ the set of states
% for which a control action is defined.
% Formally,
% $\mbox{\rm Dom}(K)$ = 
% $\{s \in S \; | \;  \exists a \; K(s, a)\}$.
%\end{definition}
%
%
% If $K$ is a controller for ${\cal S}$  
% %= ($S$, $A$, $T$)
% we denote with ${\cal S}^{(K)}$ the \emph{closed loop system}, 
% that is the LTS 
% ($S$, $A$, $T^{(K)}$), where
% $T^{(K)}(s, a, s')$ = $T(s, a, s') \wedge K(s, a)$.
%
%
%\begin{definition}%[Control Problem for LTS]
%  \label{def:ctrproblem-lts}
%A controller for a DTLHS ${\cal H}$ is a controller
%for $\mbox{\rm LTS}(\cal H)$.
%\vspace{-1mm}
\end{definition}
%\vspace*{-3mm}
%In the following 
%\marginpar{Explanation added.et
%\\ Explanation revised - is}
%\begin{definition}

We call a path $\pi$ {\em fullpath}~\cite{emerson-toplas-04}
 if either it is infinite or its last state 
$\pi^{(S)}(|\pi|)$ has no successors 
(i.e. $\mbox{\rm Adm}({\cal S}, \pi^{(S)}(|\pi|)) = \varnothing$).
We denote with ${\rm Path}(s, a)$ the set of fullpaths starting in state
$s$ with action $a$, i.e. the set of fullpaths $\pi$ such that $\pi^{(S)}(0)=s$
and $\pi^{(A)}(0)=a$.
%\end{definition}
Given a path $\pi$ in ${\cal S}$, 
%$J({\cal S},\pi,G)$ denotes the unique $n > 0$, if it exists, s.t.
%$[\pi^{(S)}(n)\in G] \land [\forall \; 0<i<n. \pi^{(S)}(i)\not\in G]$, 
%$+\infty$ otherwise.
we define $j({\cal S},\pi,G)$ as follows. If there exists $n > 0$ s.t.
$\pi^{(S)}(n)\in G$, then $j({\cal S},\pi,G) = \min\{n \;|\; n >
0 \land \pi^{(S)}(n)\in G\}$. Otherwise, $j({\cal S},\pi,G) = +\infty$.
We require $n > 0$ since
our systems are nonterminating and each controllable state (including a goal state)
must have a path of positive length to a goal state.
Taking ${\rm sup}\, \varnothing = +\infty$, %and $\inf\, \varnothing = -\infty$, 
the {\em worst case distance} %(pessimistic view) 
of a state $s$ from the goal region  $G$
is  $J({\cal S},G,s)={\rm sup} \{j({\cal S},\pi,G)~|~  \pi \in{\rm Path}(s,a), a \in {\rm
Adm}({\cal S},s)\}$.
%, being $J_s({\cal S},G,s,a)={\rm sup} \{ J({\cal S},G,\pi)~|~ \pi
%\in{\rm Path}(s,a)\}$. 
%The {\em best case distance}  (optimistic view) of a
%state $s$  from the goal region $G$ is  $J_{\rm weak}({\cal S},G,s)={\rm sup} \{
%J_w({\cal S},G,s, a)~|~ a \in {\rm Adm}({\cal S},s)\}$, being $J_w({\cal
%S},G,s,a)=\inf \{ J({\cal S},G,\pi)~|~ \pi \in{\rm Path}(s,a)\}$.

%\end{definition}
 \begin{figure}
%\vspace{-3mm}
 \begin{center}
 $
 \xymatrix@C=10mm@R=10mm{
    &
    *+=<30pt>[o][F=]{0} 
		\ar@(ld,lu)@{.>}[]^{0,1}
%	    \ar@(ul,ur)@[]^{0}
    &
	*+=<30pt>[o][F-]{1} 
		\ar@/^/@{.>}[l]^{0}
		\ar@/^/[r]^{1}		
    & 
    *+=<30pt>[o][F-]{2} 
		\ar@/^/@{.>}[l]^{0}
		\ar@/^/[d]^{1}		
	\\
	& 
	*+=<30pt>[o][F-]{4} 
		\ar@(dl,dr)@[]_{0}
 		\ar@/^/@{.>}[u]^{1}
	&
	& 
	*+=<30pt>[o][F-]{3} 
		\ar@/^/[u]^{0}
		\ar@/^/@{.>}[ul]^{1}
	&
	\\
	} 
 $
 %\vspace*{-3mm}
% \caption{Maximum ${\cal Q}$ control abstraction for DTLHS in Ex.~\ref{ex:minmax-ctr-abs}}
 \caption{The LTS ${\cal S}$ in Example~\ref{ex:mgo}.} %(Ex.~\ref{ex:minmax-ctr-abs})}
 \label{fig:lts-ex}
 %\vspace*{-4mm}
 \end{center}
 \end{figure}

%\vspace*{-1mm}
\begin{definition}
%\vspace{-1mm}
\label{def:sol}
An LTS \emph{control problem} 
%for 
%LTS 
%${\cal S}$ 
%[${\cal H}$] 
is a triple 
  $\cal P$ = $({\cal S},$ $I,$ $G)$. 
%  We suppose $G \subseteq I$ since a stabilizing 
%  controller is desired.
A {\em strong %(resp. \emph{weak}) 
solution} (or simply a solution) to %a control problem 
${\cal P}$ %= $({\cal S}$, $I$, $G)$ 
is 
a controller $K$ for ${\cal S}$, such that $I$ $\subseteq$ ${\rm dom}(K)$ and
for all $s \in {\rm dom}(K)$, 
$J({\cal S}^{(K)}, G, s)$ 
%(resp. $J_{weak}({\cal S}^{(K)}, G, s)$)
is finite.

An \emph{optimal} solution to ${\cal P}$
is a solution $K^{*}$ to ${\cal P}$ such that for all solutions
$K$ to ${\cal P}$, 
for all $s \in S$, we have 
$J({\cal S}^{(K^{*})}, G, s) \leq J({\cal S}^{(K)}, G, s)$.
%[$J_{weak}({\cal S}^{(K^{*})}, G, s) \leq J_{weak}({\cal S}^{(K)}, G, s)$].

The  \emph{most general optimal (mgo) solution}  to ${\cal P}$ 
is an optimal solution $\bar{K}$ to ${\cal P}$  such that  
for all optimal solutions $K$ to ${\cal P}$,
for all $s \in S$,
for all $a \in A$
we have $K(s, a)$ $\rightarrow$ $\bar{K}(s, a)$. %It is easy to see that
This definition is well posed (i.e., the mgo solution is unique)
and $\bar{K}$ does not depend on $I$.
%A {\em weak solution} to a control problem $({\cal S}, I, G)$ is a 
%controller $K$ for ${\cal S}$, such that $I$ $\subseteq$ $\mbox{\rm Dom}(K)$ and
%for all $s \in \mbox{\rm Dom}(K)$, 
%$J_{weak}({\cal S}^{(K)}, G, s)$ 
%is finite.
%  \begin{enumerate}
%  \item $I$ $\subseteq$ $\mbox{\rm Dom}(K)$ and
%  \item for all $s \in \mbox{\rm Dom}(K)$, $J_{\rm strong}({\cal
%        S}^{(K)}, G, s)$ is finite.
%  \end{enumerate}  
%A {\em weak solution} to a control problem $({\cal S}, I, G)$ is a 
%controller $K$ for ${\cal S}$, such that:
%  \begin{enumerate}
%  \item $I$ $\subseteq$ $\mbox{\rm Dom}(K)$ and
%  \item for all $s \in \mbox{\rm Dom}(K)$, $J_{weak}({\cal
%        S}^{(K)}, G, s)$ is finite.
%  \end{enumerate}  
%\vspace{-1mm}
\end{definition}

%\vspace{-.5cm}
\begin{example}
\label{ex:mgo}
Let ${\cal S}=(S,A,T)$ be the LTS in Fig.~\ref{fig:lts-ex}, where  
$S=\{0,1,2,3,4\}$, $A=\{0,1\}$ and the transition relation $T$ is 
defined by all arrows in the picture.
Let $I=S$ and let $G=\{ 0 \}$. 
The controller $K$ that enables all dotted arrows in the picture, 
is an mgo for the control problem $({\cal S}, I, G)$.
The controller $K'=K\setminus\{(0,1)\}$ that enables only the action 
$0$ in the state $0$, would be still an optimal solution, 
but not the most general.
The controller $K''=K\cup\{(3,0)\}$ that enables also the action $0$ in state $3$ 
would be still a solution (more general than $K$), but no more optimal. 
As a matter of fact, in this case $J({\cal S}^{(K'')},G,3)=3$, 
whereas $J({\cal S}^{(K)},G,3)=2$.  
\end{example}

\subsection{Discrete Time Linear Hybrid Systems}
\label{dths.tex}
\label{sec:dtlhs}
Many embedded control systems can be modeled as %DTLHSs.
{\em Discrete Time Linear Hybrid Sytems} (DTLHSs) %can effectively 
since they provide an uniform model both for the plant and for the control software.

%In this section we introduce the class of 
%discrete time Hybrid Systems that 
%we use as plant models, namely 
%{\em Discrete Time Linear Hybrid Systems} (DTLHSs for short).
 
%that we use to model plants
%models a Hybrid System by means of a set of continuous as well as discret 
%{\em state variables}, whose values describe the system state, 
%a set of {\em input variables}, that describe 
%a set of {\em auxiliary variables},  

%Although our MILP-based
%abstraction technique requires conjunctive predicates,
%Stemming from Prop.~\ref{predconvex.prop}
%in Prop.~\ref{predconvex.prop} 
%we model the dynamics of closed loop systems
%by using (bounded) guarded predicates.

%We will model the dynamics of closed loop systems
%by using bounded predicates. 
%As shown in Prop.~\ref{predconvex.prop},  
%any bounded predicate can be transformed into  
%an equivalent conjunctive predicate. 
%Accordingly, in Def. \ref{dths.def},
%%without loss of generality, 
%we choose guarded predicates in order to
%simplify the description of the DTLHSs considered
%in examples and experimental results,
%notwithstanding our MILP-based
%abstraction technique requires conjunctive predicates. 
 
% Many classes of piecewise affine hybrid systems have been studied in
% the literature, e.g. \cite{AHH96,lpw_cav97}.  The same holds true
% for piecewise affine discrete time hybrid systems,
% e.g. \cite{BM99,TB04,GPB05,ABCS04}. The class of systems we are
% considering is essentially the one used in \cite{watertank-cdc01}.

\begin{definition}%[DTLHS] 
\label{dths.def}
%\vspace{-2mm}
A {\em Discrete Time Linear Hybrid System} %(DTLHS) 
is a tuple ${\cal H} = (X,$ $U,$ $Y,$ $N)$ where:
%\smallskip
%\begin{itemize}
%\vspace{-.15cm}
%\item 

  $X$ = $X^{r} \cup X^{d}$ % \cup X^{b}$ % \marginpar{per essere precisi, occorrerebbe mettere anche i booleani}
  is a finite sequence of real ($X^{r}$) and %,
  discrete ($X^{d}$) %and boolean ($X^{b}$) 
  {\em present state} variables.  
  The sequence $X'$ of {\em next state} variables is 
  obtained by decorating with $'$ all variables in $X$.
%\smallskip

%\vspace{-.15cm}  
%\item 
  $U$ = $U^{r} \cup U^{d}$ % \cup U^{b}$ 
  is a finite sequence of 
  \emph{input} variables.
%\smallskip

%\vspace{-.15cm}
%\item 
  $Y$ = $Y^{r} \cup Y^{d}$ %\cup Y^{b}$ 
  is a finite sequence of
  \emph{auxiliary} variables, that  
  %Auxiliary variables 
  are typically used to
  model \emph{modes} %(e.g., from switching elements such as diodes) 
  %or \emph{uncontrollable inputs} (e.g., disturbances).
  or ``local'' variables.
%\smallskip

%\vspace{-.15cm}
%\item 
%\marginpar{N is a guarded predicate?? if not la def di bounded DTLHS
%va modificata. Commento su generalità che c'era al CAV. ET}
  $N(X, U, Y, X')$ is a conjunctive predicate 
  over $X \cup U \cup Y \cup X'$ defining the 
  {\em transition relation} (\emph{next state}) of the system.
  %$N$ is {\em deterministic} if $N(x, u, y_1, x') \land (x, u, y_2, x'')$
  %implies $x' = x''$, and {\em nondeterministic} otherwise.
%\end{itemize}
\smallskip

A DTLHS is {\em bounded} if the predicate $N$ is bounded.
%A DTLHS is {\em deterministic} if $N$ is deterministic.
%A DTLHS is {\em discrete} if all its variables are discrete.
%\vspace{-2mm}
\end{definition}

%By Prop.~\ref{predconjunctive.prop}, 
Since any bounded guarded predicate is 
equisatisfiable to a conjunctive predicate (see Sect.~\ref{subsection:predicates}), for the sake of readability we use %will use 
bounded guarded predicates to describe the transition relation of 
bounded DTLHSs. To this aim, we also clarify which variables are boolean,
and thus may be used as guards in guarded constraints.

The semantics of DTLHSs is given in terms of LTSs as follows. 

\begin{definition}
%\vspace{-2mm}
Let ${\cal H}$ = ($X$, $U$, $Y$, $N$) 
be a DTLHS.
The dynamics of ${\cal H}$ 
is defined by the Labeled Transition System 
$\mbox{\rm LTS}({\cal H})$ = (${\cal D}_X$, ${\cal D}_U$,
$\tilde{N}$) where:
% $S$ = ${\cal D}_X$, $A$ = ${\cal D}_U$ and
$\tilde{N} : {\cal D}_X \; \times \; {\cal D}_U \; \times \; {\cal D}_X \rightarrow \B$ 
is a function s.t.  $\tilde{N}(x, u, x') \equiv \exists \; y \in {\cal D}_Y \; N(x, u, y, x')$.
% Note that since $N$ is serial $\mbox{\rm LTS}({\cal H})$ is indeed
% an LTS as defined in Sect. \ref{lts.tex}.
A \emph{state} $x$ for ${\cal H}$ is a state $x$ for 
$\mbox{\rm LTS}({\cal H})$ and a \emph{run} 
(or \emph{path}) for ${\cal H}$ is
a run for $\mbox{\rm LTS}({\cal H})$. %(Sect. \ref{lts.tex}).
%A \emph{DTLHS reachability problem instance} is a triple $({\cal H}, I, G)$ with $I, G
%\subseteq {\cal D}_X$. An %\marginpar{aggiunto per poter dire che non si puo' fare il check preciso sui self-loop}
%algorithm {\em decides} the DTLHS reachability problem when, for each instance $({\cal H}, I, G)$, it returns YES iff
%there exists a finite path $\pi$ for ${\cal H}$ s.t. $\pi^{(S)}(0)\in I$ and
%$\pi^{(S)}(|\pi|) \in G$.
%\vspace{-2mm}
\end{definition}

\begin{example}
\label{ex:dtlhs}
Let T be a positive constant (sampling time).
We define the DTLHS ${\cal H}$ $=$ $(\{x\},\{u\},$ $\varnothing$, $N)$ where 
$x$ is a continuous variable, $u$ is a boolean variable, and
$N(x, u, x')$ $\equiv$
$[\overline{u} \rightarrow x' = x+(\frac{5}{4}-x)T] 
\land
[u \rightarrow x' = x+ (x - \frac{3}{2})T]$.
Since $N(\frac{5}{4},0,\frac{5}{4})$ holds, 
the infinite path $\pi_0 = \frac{5}{4},0,\frac{5}{4},0\ldots$ 
is a run in $LTS({\cal H}) = (\R, \{0,1\}, N)$.
\end{example}

\subsubsection{DTLHS Control Problem}
\label{sec:dtlhs-ctr}
\label{sec:qcp}
A DTLHS 
control problem $({\cal H}, I, G)$ is defined as the LTS
control problem ($\mbox{\rm LTS}(\cal H)$, $I$, $G$).
To manage real valued variables, in classical control theory the
concept of {\em quantization} is introduced (e.g., see
\cite{quantized-ctr-tac05}). Quantization is the process of
approximating a continuous interval by a set of integer values. 
%A
%quantized feedback control system uses two converters to translate
%continuous variables into discrete variables (AD converter) and vice
%versa (DA converter).
%
In the following we formally define a quantized feedback control
problem for DTLHSs.

%\begin{definition}
A {\em quantization function} $\gamma$
for a real interval $I=[a,b]$ is a non-decreasing   
function $\gamma:I\mapsto \Z$ such that $\gamma(I)$ is a bounded integer interval.
%We will denote $\gamma(I)$ as $\hat{I}=[\gamma(a),\gamma(b)]$.
%The \emph{quantization step} of $\gamma$, notation $\|\gamma\|$, 
%is defined as ${\rm sup}\{ \; |w-z|
%  \; | \; w, z \in I \land \gamma(w)=\gamma(z)\}$. 
%  The quantization function $\gamma$ is {\em dense} in $I$
%if $\inf\{ \; |w-z|
%  \; | \; w,z \in I\}=0$.   
%\end{definition}
%For ease of notation, 
We extend quantizations to integer intervals,
by stipulating that in such a case the quantization function
is the identity function.% (i.e. $\gamma(x) = x$).

\begin{definition}%[Quantization for DTLHS]
%\vspace{-2mm}
  Let ${\cal H} = (X, U, Y, N)$ be a DTLHS, and 
  let $W=X\cup U\cup Y$.   
  A \emph{quantization} ${\cal Q}$ for $\cal H$
  is a pair $(A, \Gamma)$, where:
%\smallskip

%  \begin{itemize}
%  \vspace{-.15cm}
%  \item
  $A$ is a %bounded rectangular 
  predicate over $W$ %=X$, $U$, and $Y$ 
  that explicitely bounds each variable in $W$. %$X$, $U$, and $Y$.
  For each $w\in W$, we denote  
  with $A_w$ its {\em admissible region} and with $A_W$ $=$ $\prod_{w\in W} A_w$.
%\smallskip
%  We denote with $A_W$ the {\em admissible region} 
%  defined by $A$ as the set of valuations that satisfy $A$.   
  
%  \vspace{-.15cm}
%  \item 
  $\Gamma$ is a set of maps $\Gamma = \{\gamma_w$ $|$
  $w \in W$ and $\gamma_w$ is a 
  quantization function for $A_w\}$. 
%  \end{itemize}
%  \vspace{-.15cm}
%\smallskip

  Let $W = [w_1, \ldots w_k]$ 
  and $v = [v_1, \ldots v_k] \in A_{W}$. 
  We write $\Gamma(v)$ for the tuple $[\gamma_{w_1}(v_1), 
  \ldots \gamma_{w_k}(v_k)]$. 
%  Finally, the \emph{quantization step} $\|\Gamma\|$ is %for $\Gamma$ is
%  defined as ${\rm sup} \{ \; \|\gamma\| \; | \; \gamma \in
%  \Gamma \}$.
%\vspace{-2mm}
\end{definition}
%Note that $\Gamma$ is univocally defined by its quantizations
%$\gamma_w$ where $w$ is a real valued variable since discrete
%variables are not affected by the quantization process (i.e. their
%quantization is always the identity function).
 
A control problem admits a \emph{quantized} solution if control
decisions can be made by just looking at quantized values. This enables
a software implementation for a controller.

%\vspace*{-0.1cm}

%\vspace*{-1mm}
\begin{definition}%[QFC $\varepsilon$-solution to a Control Problem]
%\vspace{-2mm}
  \label{def:qfc}
  Let ${\cal H} = (X, U, Y, N)$ be a DTLHS,
  ${\cal Q}=(A,\Gamma)$ be a quantization for ${\cal H}$
  and ${\cal P} = ({\cal H}, I, G)$ be a DTLHS control problem.
  A ${\cal Q}$ \emph{Quantized Feedback Control} (QFC) %strong (resp. weak) 
  solution to ${\cal P}$ is a
  %$\|\Gamma\|$ %strong (resp. weak) 
  solution $K(x, u)$ to ${\cal P}$ such that
  $
  K(x, u) = \hat{K}(\Gamma(x), \Gamma(u))
  $
  where 
  % $\hat{K}$ is a function from
  $\hat{K} : \Gamma(A_{X}) \times \Gamma(A_{U})$ $\rightarrow$ $\B$.
  % s.t.
  % $K(x, u)$ = $\hat{K}(\Gamma(x), \Gamma(u))$ is an $\varepsilon$-solution to ${\cal P}$.
  % ($\cal H$, $\Gamma^{-1}(\Gamma(I))$, $\Gamma^{-1}(\Gamma(G))$).
  % for $\cal H$ is a tuple ($\Gamma$, $\cal H$, $I$, $G$) where:
  % $\Gamma$ is a quantization for $\cal H$, 
  % $I \; \subseteq \; {\cal D}_X$ and
  % $G \; \subseteq \; {\cal D}_X$.
  % A solution to the QFC problem 
  % ($\Gamma$, $\cal H$, $I$, $G$) is a function
  % $\hat{K}$ : $\Gamma({\cal D}_{X}) \times \Gamma({\cal D}_{U})$ $\rightarrow$ $\B$
  % s.t.
  % $K(x, u)$ = $\hat{K}(\Gamma(x), \Gamma(u))$ is a solution
  % to the control problem ($\cal H$, $\Gamma^{-1}(\Gamma(I))$, $\Gamma^{-1}(\Gamma(G))$).
%\vspace{-2mm}
\end{definition}

%\input{examples.DTHS.figs.tex}

%\vspace*{-4mm}
%\begin{example}
%\label{ex:gamma-qfc-sol}
%%IM<b>
%%Let ${\cal H}=(\{x\},\varnothing,\{u\},N)$, where $D_x=[-3,3], D_u=\{0,1\}$, and $N$ defined by $N(x,\_,0,x/2)\lor N(x,\_,0,3x/2)$.
%%Let $I=D_x$ the initial region, and $G=\{0\}$.
%Let ${\cal P}$ be as in Ex. \ref{ex:ctr-dths},
%$\Gamma(x)=round(x/2)$ 
%(where $round(x) = \lfloor x \rfloor + \lfloor 2(x - \lfloor
%x\rfloor)\rfloor$ is the usual rounding function)
%and
%$\hat{K}$ as in
%Ex. \ref{example-mgo.tex}.
%Then, $\|\Gamma\|$ = 2 and
%$K(x, u)$ = $\hat{K}(\Gamma(x),\Gamma(u))$ is a $\Gamma$ QFC solution to ${\cal P}$.
%\end{example}

\begin{example}
\label{ex:ctr-dths}
Let ${\cal H}$ be the DTLHS in Ex.~\ref{ex:dtlhs}.
Let ${\cal P}$ = (${\cal H}$, $I$, $G$) be a control problem, where 
$I\equiv -2\leq x\leq 2.5$, and $G\equiv \varepsilon\leq x\leq \varepsilon$, 
for some $\varepsilon\in\R$.
%A controller may drive the system near enough to the goal %$x=0$, 
%by enabling a suitable action in such 
%a way that $x'<x$ when $x>0$ and $x'>x$ when $x<0$.
%$x'>x$ if $x\leq 0$ and $x'<x$ if $x\geq 0$. 
If the sampling time $T$ is small enough with respect to $\varepsilon$
(for example $T<\frac{\varepsilon}{10})$, 
the controller:
$
K(x,u)=(-2\leq x\leq 0\;\land\; \overline{u}) \; \lor \; 
(0\leq x\leq \frac{11}{8} \;\land\; u) \; \lor \;
(\frac{11}{8}\leq x\leq 2.5 \; \land\; \overline{u})
$ 
is a %n $\varepsilon$ 
solution to $({\cal H}, I, G)$.
Observe that any controller $K'$ such that $K'(\frac{5}{4},0)$ holds is not 
a solution, because in such a case ${\cal H}^{(K)}$ may loop forever along 
the path $\pi_0$ of Ex.~\ref{ex:dtlhs}.
%\end{example}

%\begin{example}
\label{ex:q-ctr}
%Let ${\cal P}$ be as in Example \ref{ex:ctr-dths}.
Let us consider the quantization $(A, \Gamma)$ where $A=I$
% is the predicate $-2\leq x\leq 2$ 
and $\Gamma$ = $\{\gamma_x\}$ and $\gamma_x(x)=\lfloor x\rfloor$.
The set $\Gamma(A_x)$ of quantized states is the integer interval $[-2,2]$. 
%Let $\hat{K}(\hat{x},\hat{u})=(\hat{u}=0\;\land\;\hat{x}\neq 0)\;\lor(\hat{u}=1\;\land\;\hat{x}\in
%\{0,1\})$.
%The controller $K''(x,u)=\hat{K}(\Gamma(x),\Gamma(u))$ is exactly $K$, and therefore it is 
%a QFC weak solution to ${\cal P}$.
%On the other hand, 
No %${\cal Q}$ QFC 
solution can exist, because in state $1$ either enabling action $1$ or %action  
$0$ %both $\hat{K}(1,1)$ and $\hat{K}(1,0)$ 
allows infinite loops to be potentially executed in 
the closed loop system. The controller $K$ above can be 
obtained as a quantized controller decreasing the quantization step, for example by taking 
$\tilde{\Gamma}$ = $\{\tilde{\gamma}_x\}$ where $\tilde{\gamma}_x(x)=\lfloor 8x\rfloor$. 
%\vspace{-2mm}
\end{example}

        %%\vspace{-2mm}
\subsection{Control Software Generation}
\label{sec:ctr-sftw-syn}

Quantized controllers can be computed by solving LTS control problems: 
the \qks\ control software synthesis procedure consists of 
building a suitable finite state abstraction ({\em control abstraction}) $\hat{\cal H}$ 
induced by the quantization of a plant modeled as a DTLHS ${\cal H}$, 
computing an abstraction $\hat{I}$ (resp. $\hat{G}$) of the initial (resp. goal) region $I$ 
(resp. $G$) so that any solution to the LTS control problem $(\hat{\cal H}, 
\hat{I}, \hat{G})$  
is a finite representation of a solution to $({\cal H}, I, G)$. 
In~\cite{qks-cav10}, we give a constructive sufficient condition ensuring that the 
controller computed for $\hat{\cal H}$ 
is indeed a quantized controller for ${\cal H}$. 
%Such a condition stems from the notion of {\em control abstraction}. 

%\vspace{-1mm}
\subsubsection{Symbolic Controller Synthesis}
\label{sec:ctr-syn}
Control abstractions for bounded DTLHSs are finite LTSs. 
For example, a typical quantization is the {\em uniform quantization} 
%functions 
which consists in dividing the domain of each state variable $x$ %($x_{1}, x_2$) 
into $2^{b_x}$ equal intervals, where $b_x$ is the number of
bits used by AD conversion. %This leads to an
Therefore, the abstraction of a DTLHS induced by a uniform quantization has 
$2^B$ states, where $B=\sum_{x\in X} b_x$. 
By coding states and actions as sequences of bits, 
a finite LTS can be represented as an OBDD %{\em Ordered Binary Decision Diagram} (OBDD), 
representing set of states and the transition relation by using their characteristic 
functions.
%\begin{remark}
%\label{rem:uq}
%\vspace{-2mm}
%\end{remark}

The \qks\ control synthesis procedure implements the function \fun{mgoCtr} in Alg.~\ref{strngctr.alg}, 
which adapts the algorithm presented in \cite{strong-planning-98}. 
Starting from goal states, 
the most general optimal controller is found incrementally %looping backward, 
adding at each step to the set of states $D(s)$ 
controlled so far, the {\em strong preimage} of $D(s)$, i.e. the set of states for which 
there exists at least an action $a$ that drives the system to $D(s)$, 
regardless of possible nondeterminism.

%\vspace{-2mm}
\begin{algorithm}
\caption{Symbolic Most General Optimal Controller Synthesis}
\label{strngctr.alg}
\begin{algorithmic}[1]
\REQUIRE An LTS control problem $({\cal S}, I, G)$, ${\cal S} = (S, A, T)$.
\ENSURE \fun{mgoCtr}(${\cal S}, I, G$)
%\STATE let ${\cal S} = (S, A, T)$
\STATE $K(s, a) \gets 0$, $D(s) \gets G(s)$, $\tilde{D}(s) \gets 0$
\WHILE {$D(s) \neq \tilde{D}(s)$} 
%\land \exists s\; [I(s) \land \not\exists a\; K(s, a)]$} 
  \STATE $F(s, a) \gets \exists s'\; T(s, a, s') \land \forall s'\; [T(s, a, s') \Rightarrow D(s')]$
  \STATE $K(s, a) \gets K(s, a) \lor (F(s, a) \land \not\exists a \; K(s, a))$
  \STATE $\tilde{D}(s) \gets D(s)$, $D(s) \gets D(s) \lor \exists a\; K(s, a)$
\ENDWHILE
\STATE {\bf return} $\langle \forall s\; [I(s) \Rightarrow \exists a\; K(s, a)], \exists a\; K(s, a), K(s, a)\rangle$
\end{algorithmic}
\end{algorithm}

%\begin{figure}
%\vspace{-5mm}
%\centering
%\hspace{-4mm}
%\begin{tabular}{cc}
%\begin{minipage}{0.5\columnwidth}
%%  \centering
%\hspace{-7mm}
%\includegraphics[width=1.18\textwidth]{pictures/f.pdf}
%  \vspace*{-11mm}
%  \caption{OBDD for $K$.}% (Ex.~\ref{ex:mgo}).}
%  \label{fig:obdd-k}
%  \end{minipage}
%  &
%  \hspace{-5mm}
%  \begin{minipage}{0.5\columnwidth}
%%  \centering
%\hspace{-5mm}
%\includegraphics[width=1.05\textwidth]{pictures/f-macro.pdf}
%  \vspace*{-2.7mm}
%  \caption{OBDD for $K^\ast$ (Ex.~\ref{ex:small-controller}).}
%  \label{fig:obdd-k*}
%  \end{minipage}
%\end{tabular}
%\vspace*{-1.4mm}
%\end{figure}

\begin{figure}
%%\vspace{-3mm}
\centering
%\hspace{-4mm}
\begin{tabular}{cc}
\hspace{-5mm}
\begin{minipage}{0.46\columnwidth}
%  \centering
\hspace{-7mm}
%\vspace{-1mm}
\includegraphics[width=0.9\textwidth]{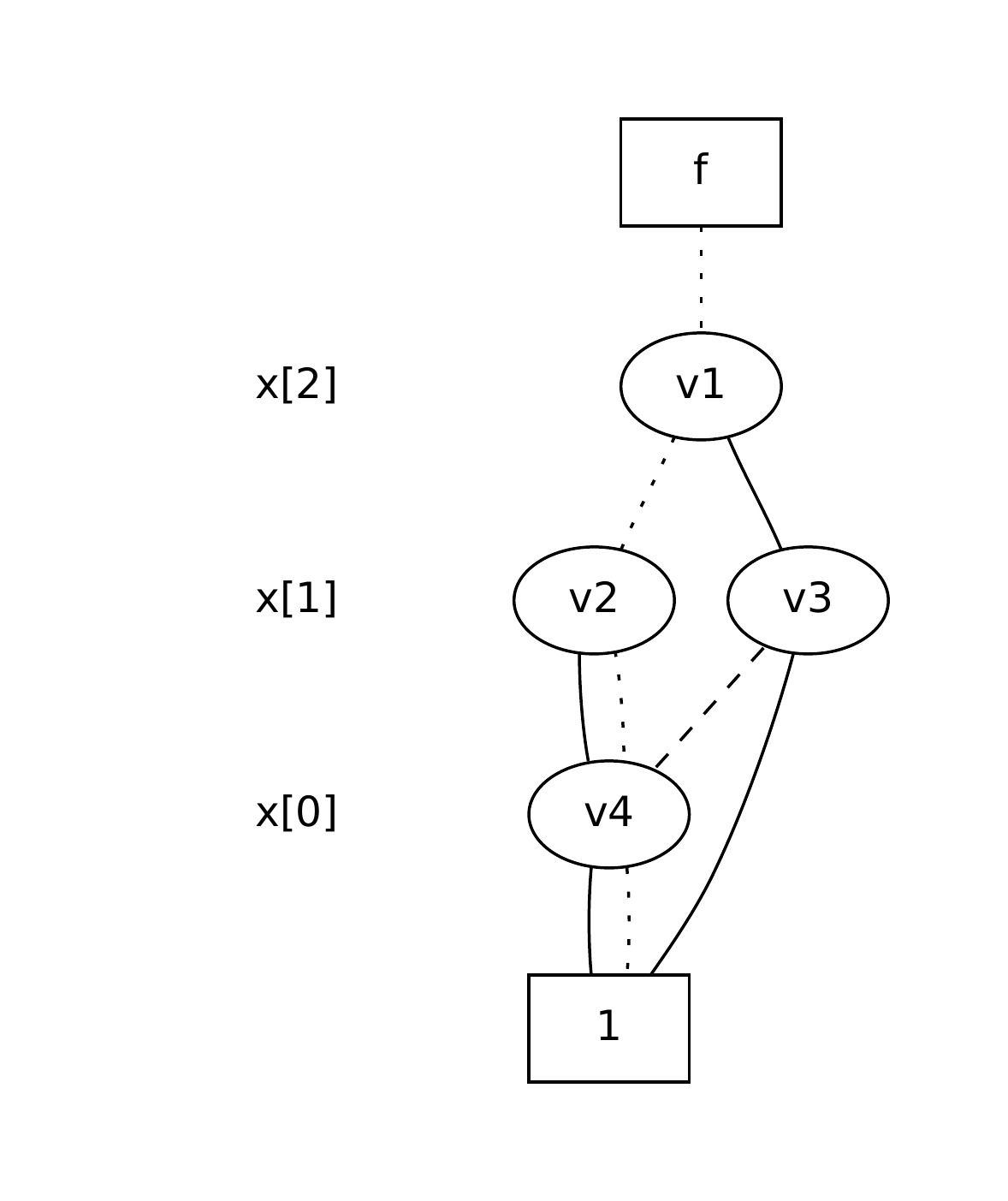}
  %\vspace*{-8mm}
  \caption{OBDD for $F$.}% (Ex.~\ref{ex:mgo}).}
  \label{fig:obdd-k}
  \end{minipage}
&
 \hspace{-4mm}
\begin{minipage}{0.54\columnwidth}
  \hspace{-4mm}
%  \vspace{6mm}
\framebox[0.9\columnwidth][c]{
\hspace{7.5mm}
\begin{minipage}{\columnwidth}
\input{c2.tex}
\end{minipage}
}
%\vspace{-3mm}
\hspace{-5mm}
\caption{C control software.}
\label{fig:c-software}
\end{minipage}
\end{tabular}	
%\vspace*{-5mm}
\end{figure}

%\input{c-software.tex}

%\vspace{-4.5mm}
\subsubsection{C Code Generation}
\label{sec:ctr-sftw-gen}
%Given an OBDD representing the controller $K(x,u)$, let $k$ %(resp. $n$) 
The output of the function \fun{mgoCtr} is an OBDD $K$ representing an mgo as 
a relation $K(x, u)$.  
Let $k$ be the number of bits used to represent the set of actions. %(resp. states). 
We are interested in a {\em control law} $F=[f_1,\ldots,f_k]$ such that 
$K(x, F(x))$ holds for all $x$~\cite{Tro98}. 
We first compute $k$ OBDDs $f_1, \ldots f_k$ representing $F$. 
%such that $K(x, f_1(x), \ldots f_k(x))$ holds. 
For any $f_i$, by replacing each OBDD node with an \texttt{\small if-} {\tt {\small then-else}} block
and each OBDD edge with a \texttt{\small goto} statement, we obtain 
a C function {\tt f\_i} that implements the boolean function represented by $f_i$.  
Therefore, the size of %the C function 
{\tt f\_i} is proportional to the number of nodes in $f_i$. 
Its WCET is proportional to the {\em height} of $f_i$, 
since any computation of {\tt f\_i} corresponds to going through a path of $f_i$. 
As a consequence, the WCET of the control software turns out to be {\em linear} in the number of 
bits of the quantization schema. 
The C function {\tt ctrLaw} is obtained by translating the $k$ OBDDs representing $F$, 
% the controller $K$ (output of \fun{mgoCtr}), 
 whereas {\tt ctrReg} is obtained by translating the OBDD representing the characteristic function of ${\rm dom}(K)$. 
 The actual code implementing control software is slightly more complicated to account for node sharing among OBDDs $f_1, \ldots, f_k$. %complicated, since it has to 
%account for OBDD sharing.
Full details about the control software generation can be found in~\cite{icsea2011}.
%\vspace{-2mm}

\begin{example}
Let ${\cal P}=({\cal S}, I, G)$ be the control problem in Ex.~\ref{ex:mgo}. 
The five states of ${\cal S}$ can be represented by three boolean variables ($x_0, x_1, x_2$). 
Taking as input ${\cal P}$, \fun{mgoCtr} computes the mgo $K$ given in Ex.~\ref{ex:mgo}.
The control law $F$ is the OBDD depicted in Fig.~\ref{fig:obdd-k}. 
%whereas the set of actions is represented by the boolean variable $u$. %represents the action.   
%\vspace{-1mm}
In Fig.~\ref{fig:c-software}, it is shown a snapshot of the control software generated for $F$. 
%the controller in Ex.~\ref{ex:mgo}. 
%\vspace{-3mm}
\end{example}

%  \input{dtlhs.tex}
%\newpage
 % \input{problem.tex}
 
  %\newpage
%\vspace{-2mm}
\section{Small Controllers Synthesis}
\label{sec:small-controllers}

%The size of control software is one of the main obstacles to overcome to make 
%model based control software synthesis viable on large  problems. 
%As explained above (2.5), the size of the control software is proportional to the size of 
%the OBDD computed by \fun{mgoCtr} (Alg.~\ref{strngctr.alg}). 
%This suggests looking for control synthesis algorithms that reduce the size of such an OBDD.
%This is the goal of our 
%control synthesis algorithm presented in Sect.~\ref{sec:control-synthesis-algo}.
%Intuitively, we reduce the number of OBDD nodes by increasing node sharing.
%This is be done by looking for control laws that are constant on large regions of 
%the state space. 
%In our setting, this means considering controllers that enable the same action in 
%as large as possible regions of the state space. 
%Note that changing the control synthesis algorithm does not change the 
%WCET of the generated control software since it only depends on the number of quantization bits
%(Sect. 2.5). 
%\smallskip

Within the framework defined in the previous section, 
when finer (i.e. with more bits) quantization schemas are considered, 
better controllers are found, in terms of set-up time and ripple 
(see Sect.~\ref{expres.tex}).
On the other hand, the exponential growth of control software size 
is one of the main obstacles to overcome in order to make 
model based control software synthesis viable on large  problems. 
As explained in Sect.~\ref{sec:ctr-sftw-gen}, 
the size of the control software is proportional to the size of 
the OBDD computed by the function \fun{mgoCtr} in Sect.~\ref{sec:ctr-syn}. 
%as faster as possible.
To reduce the number of nodes of such an OBDD, we devise a heuristic aimed at  
increasing OBDD node sharing by looking for control laws that are 
constant on large regions of the state space.

While optimal controllers implement smart control strategies that in each state try to find 
the best action to drive the system to the goal region, the function \fun{smallCtr} in 
Sect.~\ref{sec:control-synthesis-algo}
 looks for more ``regular''  
controllers that enable the same action in 
as large as possible regions of the state space. 

Finally, note that changing the control synthesis algorithm does not change the 
WCET of the generated control software since it only depends on the number of quantization bits
(Sect.~\ref{sec:ctr-sftw-gen}).

\subsection{Control Synthesis Algorithm}
\label{sec:control-synthesis-algo}

Our controller synthesis algorithm is shown in Alg.~\ref{alg:smallctr4}. 
To obtain a succinct controller, 
the function \fun{smallCtr} modifies the \fun{mgoCtr} preimage 
computation of set of states $D$   
by {\em finding maximal regions} of 
states from which the system reaches $D$ in {\em one or more steps}  
by repeatedly performing the {\em same action}.
This involves finding at each step a family of fixpoints: for each action $a$, 
$E(s,a)$ is the maximal set of states from which $D$ is reachable 
by repeatedly performing the action $a$ only. 

The function \fun{smallCtr}(${\cal S}, I, G$) %in Alg.~\ref{alg:smallctr4} 
computes a solution $K$ 
to the control problem (${\cal S}, I, G$) (Theor.~\ref{thm:correctness}),   
such that ${\rm dom}(K)$ is {\em maximal} with respect to any other 
solution (Theor.~\ref{thm:maximal}). 
% in order to obtain a 
%succinct controller, i.e. representable with an OBDD (and hence C code) 
%much smaller than the one needed for the mgo. 

In Alg.~\ref{alg:smallctr4} $K(s,a)$ denotes the OBDD that represents the controller computed 
so far, $D(s)$ the OBDD that represents its domain, and $\tilde{D}(s)$ the domain 
of the controller computed at the previous iteration. 
The computation starts by initializing $K(s,a)$ and its domain $D(s)$ 
to the empty OBDD, that corresponds to the always undefined function 
and the empty set (line \ref{line:init}). 

At each iteration of the outer loop (lines~\ref{line:repeat}--\ref{line:until}), 
a target set of states $O(s)$ is considered (line~\ref{line:init-target}): 
$O(s)$ consists of goal states $G(s)$ and the set $D(s)$ of already controlled states. 
%Usually, in real control problems, we have that $G(s)\subseteq D(s)$ after few iterations. 
The inner loop (lines \ref{line:start-inner-loop}--\ref{line:end-inner-loop}) computes, 
for each action $a$, the maximal set of states $E(s,a)$ 
that can reach the target set $O(s)$ by repeatedly performing the 
action $a$ only. 
For any action $a_0$, $E(s,a_0)$ is the mgo of the control problem
$({\cal S}', I, O)$, where the LTS ${\cal S}'=(S, \{a_0\}, T')$ is obtained  
by restricting the dynamics of ${\cal S}$ to the action $a_0$.
%, i.e. $T'(s,a,s')=T(s,a,s')\,\land\,a=\tilde{a}$.

After that, $K$ is updated by adding to it state-action pairs in $E(s,a)$. 
Instead of simply computing $K(s,a)\!\gets\!K(s,a)\lor E(s,a)$, to keep the controller smaller, 
function \fun{smallCtr} avoids to add to $K$ possible 
intersections between any pair of sets $E(s,a)$ and $E(s,b)$ 
for $a\neq b$ (line \ref{line:update-K}).
As a consequence, the resulting controller $K$ is a control law, i.e. it  
enables just one action in a given state $s$. 

The order in which the loop in lines~\ref{line:for-start}--\ref{line:for-end} 
enumerates the set of actions gives priority to actions that are considered 
before. Let $a_0, a_1, \ldots, a_n$ be the sequence of actions as enumerated by 
the {\bf for} loop. If %, at some iteration, 
there exists at least one action $a$ such that $E(s,a)$ holds, 
then we will have that $K(s,a_k)$ holds only for a certain $a_k$ such that 
$k=\min \{i~|~E(s,a_i)\}$.
In many control problems, this is useful as it allow one to give priority to 
some actions, e.g. in order to prefer ``low power'' actions. 
%It is easy
%to extend our approach to deal with such an issue, by simply considering actions in a given 
%order in the {\bf for} loop.

The computation ends when no new state is added to the controllable region, 
i.e. when $D(s)$, %the domain of the controller computed so far,    
is the same as  $\tilde{D}(s)$.
%, the domain of the controller computed at the previous iteration. 

%\vspace{-2mm}
%\newpage
\begin{algorithm}[ht]
\caption{Symbolic Small Controller Synthesis} %\\ \hspace*{1cm}
%for an LTS control problem}
\label{alg:smallctr4}
\begin{algorithmic}[1]
\REQUIRE LTS control problem $({\cal S}, I, G)$, with LTS ${\cal S} = (S, A, T)$
\ENSURE \fun{smallCtr}(${\cal S}, I, G$)
\STATE $K(s, a) \gets 0$, $D(s) \gets 0$
	\label{line:init} 
\REPEAT 
	\label{line:repeat}
  \STATE $O(s)\gets D(s)\;\lor\; G(s)$, $E(s,a)\gets 0$
  	\label{line:init-target}
%  \STATE $E(s,a)\gets \exists s' T(s, a, s') \land \forall s'\; [T(s, a, s') \Rightarrow O(s')]$
%  \STATE 
  	\label{line:init-E}
  %\STATE $\tilde{E}(s,a)\gets 0$
  	%\label{line:init-target}
%  \WHILE {$E(s,a) \neq \tilde{E}(s,a)$} 
    \REPEAT
    \label{line:start-inner-loop}
    \STATE \hspace*{-1.3mm}$F(s, a) \!\gets\! \exists s' T(s, a, s') \!\land\! \forall s' [T(s, a, s')\! \Rightarrow\! E(s',a) \!\lor\! O(s')]$ 
     \label{line:Fgets}
%    \STATE \hspace*{1cm}$\forall s'\; 
    \STATE \hspace*{-1.3mm}$\tilde{E}(s,a)\gets E(s,a)$, $E(s,a)\gets E(s,a) \,\lor \, F(s,a)$
     \label{line:Egets}
%  \ENDWHILE
   \UNTIL{$E(s,a) = \tilde{E}(s,a)$}   
     \label{line:end-inner-loop}
  \FORALL{$\tilde{a} \in A$}
     \label{line:for-start}
%    \STATE $F(s, a) \gets \not\exists a \; K(s,a)$
%    \STATE $L(s, a) \gets E(s,\tilde{a}) \land F(s, a)$
%    \STATE $K(s, a) \gets K(s, a) \lor L(s, a)$
    \STATE $K(s, a) \gets K(s, a) \lor (E(s,a) \land a = \tilde{a}\,\land \not\exists b \; K(s,b))$
	\label{line:update-K}
  \ENDFOR
    \label{line:for-end}
  \STATE $\tilde{D}(s) \gets D(s)$, $D(s) \gets D(s) \lor \exists a K(s, a)$
\UNTIL{$D(s) = \tilde{D}(s)$}
	\label{line:until}
\STATE {\bf return} $\langle \forall s\; [I(s) \Rightarrow \exists a K(s, a)], \exists a K(s, a), K(s, a)\rangle$
     \label{line:return}
\end{algorithmic}
\end{algorithm}

%\vspace{-4mm}
\begin{example}
\label{ex:small-controller}

\sloppy 

Let ${\cal P}$ be the control problem described in Ex.~\ref{ex:mgo}.
The first iteration of Alg.~\ref{alg:smallctr4} 
computes the predicate $E(s,a)$ that holds on the set 
$\{(0,0), (0,1), (1,0), (2,0), (3,0), (4,1)\}$, that is  
$E(s,a)=E(s,0)\,\lor\,E(s,1)$, where the set of pairs that satisfies   
$E(s,0)$ is $\{(0,0), (1,0), (2,0), (3,0)\}$
and the set of pairs that satisfies  
$E(s,1)$ is $\{(0,1), (4,1)\}$. 
Depending on the order in which the {\bf for} loop in 
lines~\ref{line:for-start}--\ref{line:for-end} 
enumerates the set of actions, in the state 0 the resulting controller $K^*$ 
enables the action 0 ($K^*(s,a)=E(s,0)\cup (E(s,1)\setminus\{(0,1)\})$) 
or the action 1 ($K^*(s,a)=E(s,1)\cup (E(s,0)\setminus\{(0,0)\})$).
Observe that, in any case, $K^*$ is not optimal. 
An optimal controller would enable the transition $T(3,1,1)$ rather than $T(3,0,2)$
(see Ex.~\ref{ex:mgo}).

\fussy 

The OBDD representing the control law $F$ such that $K^*(x, F(x))$ holds, 
is depicted in Fig.~\ref{fig:obdd-k*}. 
It has 3 nodes, instead of the 4 nodes required for the OBDD representation 
of the control law (Fig.~\ref{fig:obdd-k}) obtained from the controller $K$  given in Ex.~\ref{ex:mgo}. Accordingly, the corresponding C code in Fig.~\ref{fig:c-software-macro} has 
3 {\tt {\small if-then-else}} blocks, instead of the 4 in the C code of Fig.~\ref{fig:c-software}. 
%\vspace{-2mm}
\end{example}

\begin{remark}
\label{remark:action-switches}
Let $\pi$ $=$ $s_0,a_0,$ $s_1,a_1,$ $\ldots,$ $a_{n-1},s_n$ %,a_n,$ %$s_{n+1}$, 
be a path. An {\em action switch} in $\pi$ occurs whenever $a_i\not=a_{i+1}$.
%the {\em number of switches} of $\pi$  
%, notation $\sharp\pi$, 
%is $|\{a_i~|~a_i\not=a_{i+1}\}_{i\in[n]}|$. 
Controllers generated by Alg.~\ref{alg:smallctr4} 
implement control strategies with a very low number of switches. 
In many systems this is a desirable property. 
A ``switching optimal'' control strategy cannot be, however, implemented 
 by a {\em memoryless state-feedback} control law. 
 As an example, take again the control problem ${\cal P}$ described in Ex.~\ref{ex:mgo}. 
 The controller defined by $E(s,a)$ in Ex.~\ref{ex:small-controller} 
 contains all switch optimal paths. However, to minimize the number of switches along the paths 
 going through state $0$, a controller should enable action $0$ when coming from state 1,
 action $1$ when coming from $4$, and repeat the last action ($0$ or $1$) when the system is 
 executing the self-loops in state $0$. In other words, only a {\em feedback controller with memory} 
 can implement this control strategy. 
\end{remark}

\begin{figure}
%\vspace{-3mm}
\centering
%\hspace{-4mm}
\begin{tabular}{cc}
\hspace{-5mm}
\begin{minipage}{0.46\columnwidth}
%  \centering
\hspace{-7mm}
\includegraphics[width=1\textwidth]{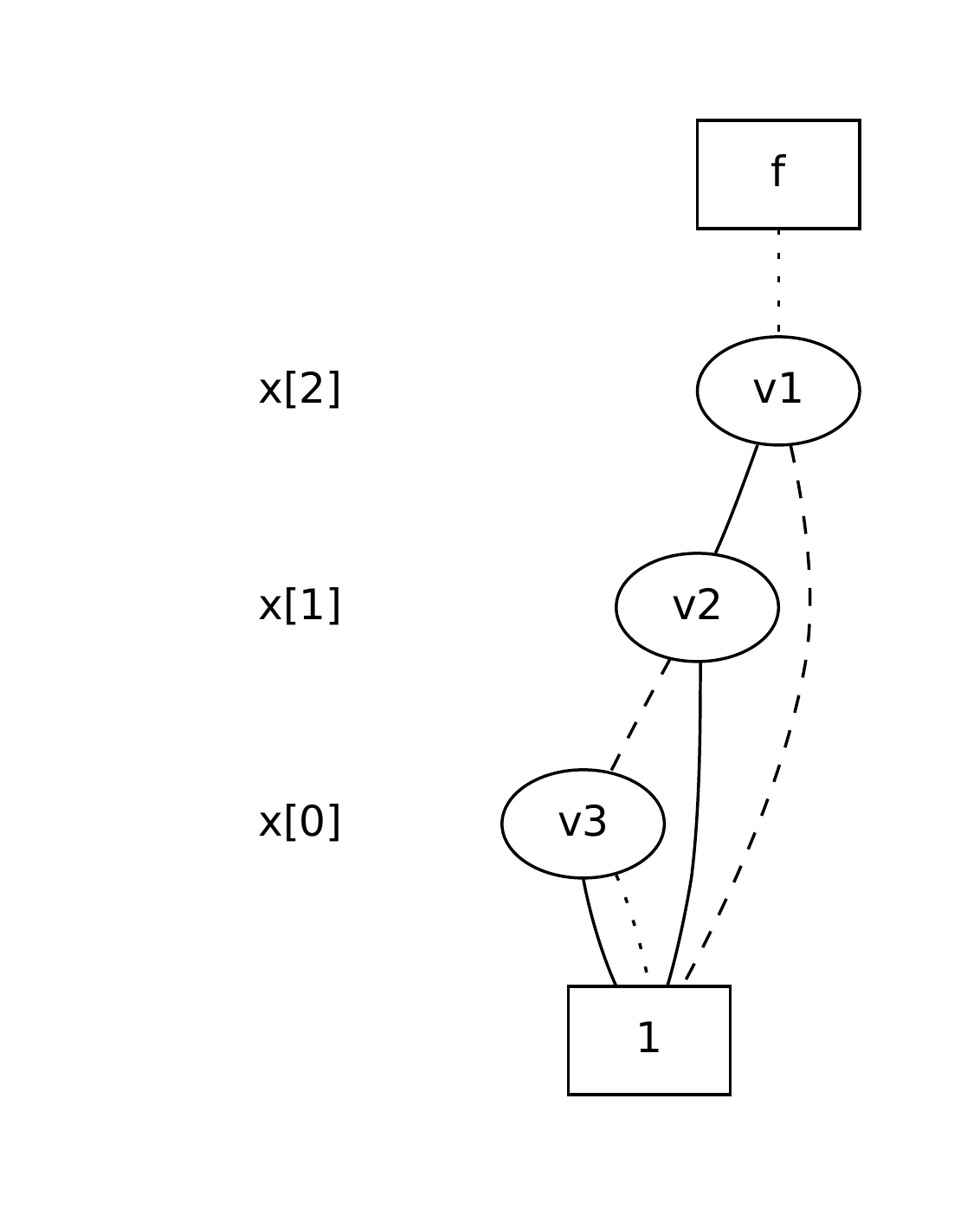}
 % \vspace*{-8mm}
  \caption{OBDD for $F^*$.}% (Ex.~\ref{ex:mgo}).}
  \label{fig:obdd-k*}
  \end{minipage}
&
 \hspace{-4mm}
\begin{minipage}{0.54\columnwidth}
  \hspace{-4mm}
 % \vspace{-6mm}
\framebox[1.1\columnwidth][c]{
\hspace{-1.6mm}
\begin{minipage}{\columnwidth}
\input{c3.tex}
\end{minipage}
}
\vspace{3mm}
\hspace{-5mm}
\caption{C code for $F^*$.}
\label{fig:c-software-macro}
\end{minipage}
\end{tabular}	
%\vspace*{-5mm}
\end{figure}

\subsection{Synthesis Algorithm Correctness\\ and Completeness}
\label{sec:algo-correctness}
In the following, we establish the correctness of Alg.~\ref{alg:smallctr4}, 
by showing that the controller computed by \fun{smallCtr} is indeed a 
solution to the control problem given as input (Theor.~\ref{thm:correctness}), 
and its completeness, in the sense that the domain 
of the computed controller is {\em maximal} with respect to the domain of 
any other solution (Theor.~\ref{thm:maximal}). 

%\vspace{-2mm} 
\begin{theorem}
\label{thm:correctness}
Let ${\cal S}=(S,A,T)$ be an LTS, and $I,G\subseteq S$ be two sets of states.
If $\fun{smallCtr}(${\cal S}, I, G$)$ 
returns the tuple $\langle \mbox{\sc True},$ $D,$ $K\rangle$, 
then $K$ is a solution to the control problem $({\cal S}, I, G)$.
%\vspace{-2mm}
\end{theorem}
\begin{proof}
If $\fun{smallCtr}(${\cal S}, I, G$)$ 
returns the tuple $\langle \mbox{\sc True},$ $D,$ $K\rangle$, 
clearly $I\subseteq {\rm dom}(K)$ (see Alg.~\ref{alg:smallctr4}, line\ref{line:return}). 
We have to show that, for all $s\in {\rm dom}(K)$, 
$J({\cal S}^{(K)},G,s)$ is finite. 

First of all, we show that %for any predicate $E(s,a)$ computed 
%during the execution of \fun{smallCtr}, 
at the end of the inner {\bf repeat} loop of \fun{smallCtr} (lines \ref{line:start-inner-loop}--\ref{line:end-inner-loop}), 
if $E(s,a)$ holds, then we have that 
$J({\cal S}^{(E)},$ $O,$ $s)$ is finite. We proceed by induction on the number of 
iteration of the inner {\bf repeat} loop.
Denoting with $F_i(s,a)$ the predicate $F(s,a)$ computed in line~\ref{line:Fgets} 
during the $i$-th iteration, 
we will show that if $F_n(s,a)$ holds, then $J({\cal S}^{(E)},$ $O,$ $s)=n$.
If $F_1(s,a)$ holds, then for all $s'$ such that $T(s,a,s')$, $s'$ belongs to $O$, 
and hence $J({\cal S}^{(E)},$ $O,$ $s)=1$. 
Along the same lines, %and resting on how $E$ is incremented in line~\ref{line:Egets}, 
if $F_{n+1}(s,a)$ holds, then $J({\cal S}^{(E)},$ $F_{n},$ $s)=1$, 
and by applying induction hypothesis, $J({\cal S}^{(E)},$ $O,$ $s)=n+1$. 
As for termination, we have that if $\tilde{E}(s,a)\neq E(s,a)$
then at least one new state has been included in $E(s,a)$. Thus the function 
$|S|-|{\rm dom}(E)|$ is strictly positive and strictly decreasing at each iteration. 
%and positive if the loop guard holds. 

The outer {\bf repeat} loop behaves in a similar way. 
Denoting with $E_i(s,a)$ the predicate $E(s,a)$ computed in line~\ref{line:init-E}  
during the $i$-th iteration, 
if $s\in {\rm dom}(K)$, then $E_i(s,a)$ holds for some $i$ and some $a$. 
We prove the statement of the theorem by induction on $i$.
If $i\!=\!1$, we have that $O(s)\!=\!G(s)$ and that 
$J({\cal S}^{(E_1)},$ $O,$ $s)$ is finite, and hence 
trivially $J({\cal S}^{(K)},$ $G,$ $s)$ is finite.
If $i>1$, then we have that $J({\cal S}^{(E_i)},$ ${\rm dom}(E_{i-1}),$ $s)$ 
is finite. Since, by inductive hypothesis, also $J({\cal S}^{(E_{i-1})},$ $O,$ $s)$ 
is finite, we have that $J({\cal S}^{(K)},$ $G,$ $s)$ $\leq$ $J({\cal S}^{(E_i)},$ ${\rm dom}(E_{i-1}),$ $s)$ + $J({\cal S}^{(E_{i-1})},$ $O,$ $s)$ is finite. 
%\vspace{-2mm}
\end{proof}

\begin{theorem}
\label{thm:maximal}
Let ${\cal S}=(S,A,T)$ be an LTS, and $I,G\subseteq S$ be two sets of states.
If $\fun{smallCtr}(${\cal S}, I, G$)$ 
returns the tuple $\langle \mbox{\sc True},$ $D,$ $K\rangle$, 
then $D={\rm dom}(K)$ is the maximal controllable region, i.e. 
for any other solution $K^*$ to the control problem (${\cal S}, I, G$) 
we have ${\rm dom}(K^*)\subseteq{\rm dom}(K)$.
%\vspace{-2mm}
\end{theorem}
\begin{proof}
Let ${\rm dom}_n(K)=\{s~|~J({\cal S}^{(K)},s,G)=n\}$. 
We will show by induction that, for all $n$, ${\rm dom}_n(K^*)\subseteq{\rm dom}(K)$.

$(n=1)$ Let $s\in{\rm dom}_1(K^*)$. Then $\mbox{\rm Adm}({\cal S}, s)\not=\varnothing$ 
and there exists at least one action $a\in{\rm Adm}({\cal S}, s)$ such that 
$K^*(s,a)$ holds. Thus, for all $s'$ such that $T(s,a,s')$ we have that $s'\in G$.
But this means that $F(s,a)$ holds 
(Alg.~\ref{alg:smallctr4}, line~\ref{line:Fgets}) 
and therefore $K(s,a)$ holds. Hence $s\in {\rm dom}(K)$. 

$(n>1)$ Let $s\in{\rm dom}_n(K^*)$. Then $\mbox{\rm Adm}({\cal S}, s)\not=\varnothing$ 
and there exists at least an action $a\in{\rm Adm}({\cal S}, s)$ such that  
$K^*(s,a)$ holds. Thus, for all $s'$ such that $T(s,a,s')$ we have that $s'\in {\rm dom}_{n-1}(K^*)$. 
By inductive hypothesis, ${\rm dom}_{n\!-\!1}(K^*)\!\subseteq\!{\rm dom}(K)$.
Therefore, for all $s'$ such that $T(s,a,s')$ we have that $s'\!\in\! {\rm dom}(K)$.
Let us suppose that $s\not\in{\rm dom}(K)$. But this implies that $\mbox{\rm Img}({\cal S}, s, a)$ 
$\not\subseteq {\rm dom}(K)$, otherwise Alg.~\ref{alg:smallctr4} would not terminated 
before adding $s$ to $E(s,a)$ at some iteration. This leads to a contradiction, 
because $\mbox{\rm Img}({\cal S}, s, a)$ $\subseteq$ ${\rm dom}_{n-1}(K^*)\subseteq{\rm dom}(K)$.
\end{proof}

  \section{Experimental Results}
\label{expres.tex}% \Blue{Igor, Federico}}
\label{sec:expres}
%\Red{Describe Experimental results for buck dc-dc converter as well as inverted pendolum.}
%\vspace*{-0.1cm}
In this section we present our experiments that aim at
evaluating the effectiveness of our control software synthesis technique. 
We mainly evaluate the control software size reduction and the impact on 
other non-functional control software requirements such as set-up time (optimality) and 
ripple.

We implemented \fun{smallCtr} in the C programming language using 
the CUDD~\cite{CUDD} package for OBDD based computations.
The resulting tool, \qkssc, extends the tool \qks\ by adding the possibility of 
synthesising control software (step 2 in Fig.~\ref{fig:cssf}) 
by using \fun{smallCtr} instead of the mgo 
controller synthesis \fun{mgoCtr}.
 
In Sect.~\ref{sec:inverted-pendulum-model} and ~\ref{sec:multi-input-buck-model}
we will present the DTLHS models of the inverted pendulum and the multi-input 
buck DC-DC converter, on which our experiments focus.
In Sect.~\ref{expres.tex.setting} we give the details of the experimental setting, and 
finally, in Sect.~\ref{sec:exp-disc}, we discuss experimental results.

%\subsection{Benchmark Models}
%\newpage
%\begin{example}
\subsection{The Inverted Pendulum as a DTLHS}
\label{ex:dths}
\label{sec:inverted-pendulum-model}

\begin{figure}
  \centering
  %\vspace{-3mm}
  \includegraphics[width=0.75\columnwidth]{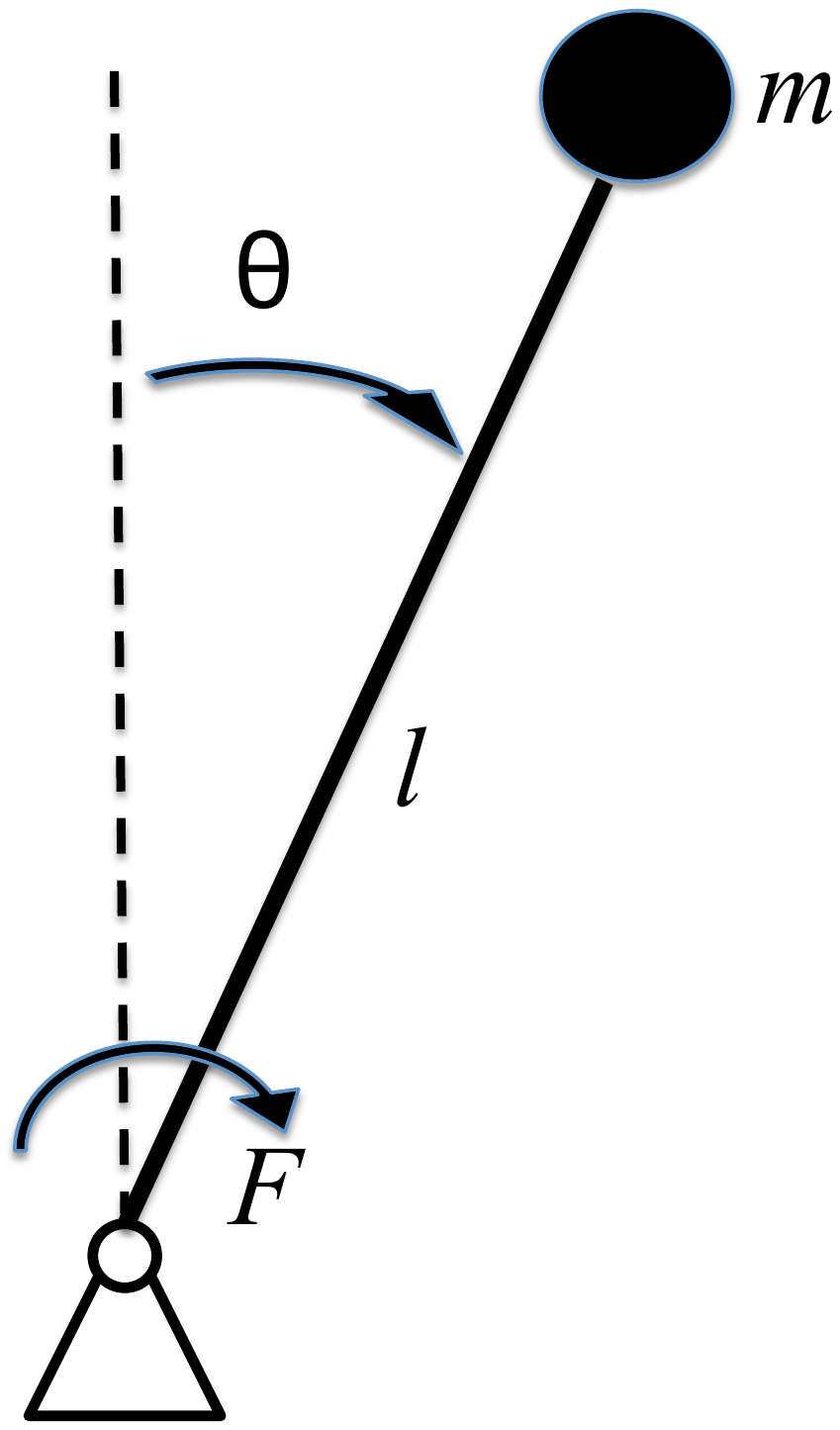}
  %\vspace*{-9.7mm}
  \caption{Inverted Pendulum with Stationary Pivot Point.}
  \label{fig:invpend}
  % \caption{Inverted Pendulum with Stationary Pivot Point and Massless Rod}
%\vspace*{-4.3mm}
\end{figure}

The inverted pendulum~\cite{KB94} (see Fig.~\ref{fig:invpend}) is modeled by taking the angle $\theta$ and the angular velocity $\dot{\theta}$ as 
state variables. The input of the system is the torquing force $u\cdot F$, that can influence 
the velocity in both directions. Here, the variable $u$ models the direction and the 
constant $F$ models the intensity of the force. 
Differently from~\cite{KB94}, we consider the problem of finding a discrete controller, whose decisions may be  only ``apply the force clockwise'' ($u=1$), ``apply the force counterclockwise'' ($u=-1$)'', or ``do nothing'' ($u=0$). 
The behaviour of the system depends on the pendulum mass $m$, the length of the 
pendulum $l$, and the gravitational acceleration $g$. Given such parameters, 
the motion of the system is described by the differential equation 
$\ddot{\theta} = \dfrac{g}{l} \sin \theta + \dfrac{1}{m l^2} u F$.
%\begin{equation}
%  \label{eq:pendmotion}
%\end{equation}
%\\[\smallskipamount]
In order to obtain a state space representation, we consider the following normalized system,  
where $x_1$ is the angle $\theta$ and $x_2$ is the angular speed $\dot{\theta}$:
%\vspace{-.2cm}
\begin{equation}
  \label{eq:pendmotion2}
\left\{\begin{array}{l}
    \dot{x}_1 = x_2 \\
    \dot{x}_2 = \dfrac{g}{l} \sin x_1 + \dfrac{1}{m l^2} u F
  \end{array}
\right.
\end{equation}
The discrete time model obtained from the equations in (\ref{eq:pendmotion2}) 
by introducing a constant $T$ that models the sampling time is: 
%\vspace{-.1cm}
\[(x'_1 = x_1 + T x_2)\,\land\,(x'_2 = x_2 + T {g \over l} \sin x_1 + T {1 \over m l^2} u F)
\]
that is not linear, 
as it contains the function $\sin x_1$. 
A linear model can be found by under- and over-approximating the non linear function 
$\sin x$. 
In our experiments (Sect.~\ref{sec:expres}), we will consider the linear model 
obtained as follows. 

First of all, in order to exploit sinus periodicity, we consider the equation $x_1 = 2 \pi y_{k} +
y_{\alpha}$, where $y_{k}$ represents the period in which $x_1$ lies and
$y_{\alpha} \in [-\pi, \pi]$\footnote{In this section we write $\pi$ for a rational approximation of it.} represents the actual $x_1$ inside a given period.
Then, we partition the interval $[-\pi, \pi]$ in four intervals: 
$I_1$ $=$ $\left[-\pi, -\dfrac{\pi}{2}\right]$, 
$I_2$ $=$ $\left[-\dfrac{\pi}{2},0\right]$, $I_3$ $=$ $\left[0, \dfrac{\pi}{2}\right]$, 
$I_4$ $=$ $\left[\dfrac{\pi}{2},\pi\right]$. 
In each interval $I_i$ ($i\in[4]$), we consider two linear functions $f_i^+(x)$ and 
and $f_i^{-}(x)$, such that for all $x\in I_i$, we have that $f_i^-(x)\leq \sin x\leq f_i^+(x)$.
As an example, $f_1^{+}(y_{\alpha})$ $=$ $-0.637 y_{\alpha} - 2$ and 
$f_1^{-}(y_{\alpha})$ $=$ $-0.707 y_{\alpha} - 2.373$. 

Let us consider the set of fresh continuous variables $Y^r$ $=$ $\{y_\alpha,$ $y_{\sin}\}$ 
and the set of fresh discrete variables $Y^d=\{y_k, y_q,$ $y_1, y_2,$ $y_3, y_4\}$, with $y_1, \ldots, y_4$ being boolean variables. 
The DTLHS model ${\cal I}_F$ for the inverted pendulum is the tuple $(X,U,$ $Y,N)$, 
where $X=\{x_1, x_2\}$ is the set of continuous state variables, 
$U=\{u\}$ is the set of input variables, $Y=Y^r\,\cup\, Y^d$ is the set of auxiliary variables, and the transition 
relation $N(X,U,Y,X')$ is the following predicate: 
%\vspace{-.25cm}
\[
\begin{array}{l}
(x'_1 = x_1 + 2\pi y_q + T x_2)	\,\land\, (x'_2 = x_2 + T\dfrac{g}{l} y_{\sin} + T\dfrac{1}{m l^2} u F)
\\[\medskipamount]
\hspace*{.3cm} \land \,\bigwedge_{i\in[4]} y_i\rightarrow f_i^-(y_\alpha)\leq y_{\sin}\leq f_i^+(y_\alpha)
\\[\medskipamount]
\hspace*{.3cm}\land \, \bigwedge_{i\in[4]} y_i\rightarrow y_\alpha \in I_i 
\land \sum_{i\in[4]} y_i \geq 1
\\[\medskipamount]
\hspace*{.3cm}\land \; x_1=2\pi y_k + y_\alpha \;\land\; -\pi \leq x_1' \leq \pi
%\\[\medskipamount]
\end{array}
\]
Overapproximations of the system behaviour increase system nondeterminism.
Since ${\cal I}_F$ dynamics overapproximates the dynamics of the non-linear model,  
the controllers that we synthesize are inherently {\em robust}, that is they meet 
the given closed loop requirements %(safety as well as liveness) 
{\em notwithstanding} nondeterministic small \emph{disturbances} such as
variations in the plant parameters.
Tighter overapproximations of non-linear functions makes finding a controller easier, 
whereas coarser overapproximations makes controllers more robust.
%\end{example}

%\begin{example}
%\label{ex:pendulum-goal}
The typical goal for the inverted pendulum %in Example~\ref{ex:dths} 
is to turn the pendulum steady 
to the upright position, starting from any possible initial position, within a given speed interval. 
%In our experiments, the goal region is defined by the predicate $G(X)\equiv (-\rho\leq x_1\leq\rho)\,\land\,(-\rho\leq x_2\leq\rho)$, where $\rho\in\{0.05, 0.1\}$, and the initial 
%region is defined by the predicate $I(X)\equiv(-\pi\leq x_1\leq \pi)\,\land\,(-4\leq x_2\leq 4$). 
%\end{example}

%\newpage
\subsection{Multi-input Buck DC-DC Converter}
\label{sec:multi-input-buck-model}

\begin{figure}
  %\centering
  \hspace*{-2.7cm}\includegraphics[scale=0.55]{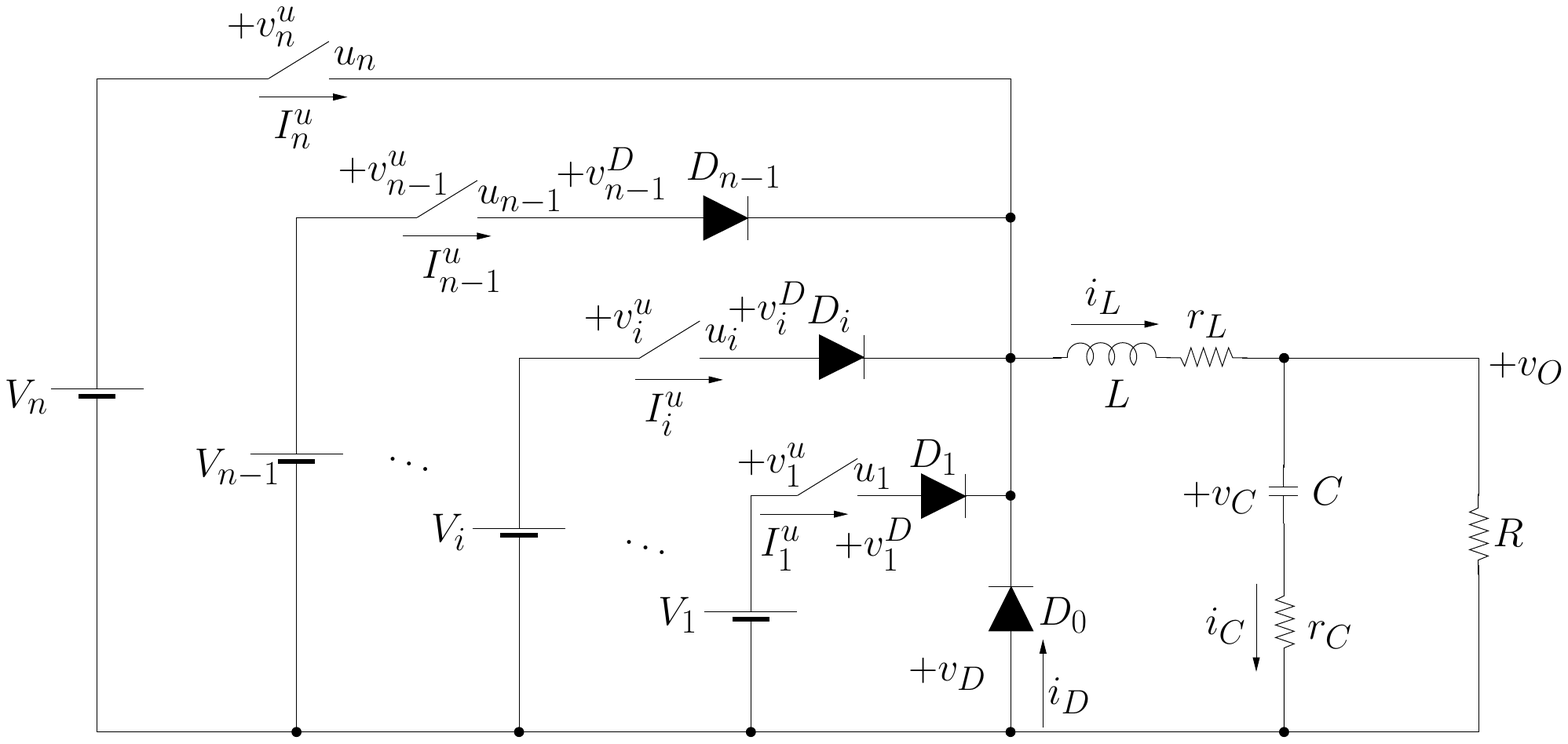}
  \vspace*{-5cm}
  \caption{Multi-input Buck DC-DC Converter.}
  \label{fig:mibdcdc}
%\vspace*{-0.2cm}
\end{figure}

The {\em multi-input} buck DC-DC converter \cite{multin-buck-dcdc-2010} in Fig.~\ref{fig:mibdcdc} is a mixed-mode analog
circuit converting the DC input voltage ($V_i$ in Fig.~
\ref{fig:mibdcdc}) to a desired DC output voltage ($v_O$ in Fig.~
\ref{fig:mibdcdc}).
As an example, buck DC-DC converters are used off-chip to scale down
the typical laptop battery voltage (12-24) to the just few volts
needed by the laptop processor (e.g. \cite{fuzzy-dc-dc-1996}) as well
as on-chip to support \emph{Dynamic Voltage and Frequency Scaling}
(DVFS) in multicore processors
(e.g. \cite{gigascale-integration-07}).
%Because of its widespread use,
%control schemas for buck DC-DC converters have been widely studied
%(e.g. see
%\cite{gigascale-integration-07,fuzzy-dc-dc-1996}). %,time-optimal-dc-dc-2008}).
The typical software based approach (e.g. see \cite{fuzzy-dc-dc-1996}) is to control
the switches $u_1, \ldots, u_n$ in Fig.~\ref{fig:mibdcdc} (typically implemented with a
MOSFET) with a microcontroller.

In such a converter %(Fig.~\ref{fig:mibdcdc}),
there are $n$ power supplies with voltage values $V_1, \ldots, V_n$, $n$
switches with voltage values $v_1^u, \ldots, v_n^u$ and current values $I_1^{u},
\ldots, I_n^{u}$, and $n$ input diodes $D_0, \ldots, D_{n - 1}$ with voltage
values $v_0^D, \ldots, v_{n - 1}^D$ and current $i_0^D, \ldots, i_{n - 1}^D$ (in
the following, we will write $v_D$ for $v_0^D$ and $i_D$ for $i_0^D$). 

%The circuit in Fig.~\ref{fig:mibdcdc} can be modeled as a DTLHS ${\cal
%  H}$ = ($X$, $U$, $Y$, $N$).
The circuit state variables are $i_L$ and $v_C$.  However we can
also use the pair $i_L$, $v_O$ as state variables in the DTLHS 
model since there is a linear relationship between $i_L$, $v_C$ and
$v_O$, namely: $v_O \; = \; \frac{r_C R}{r_C + R} i_L + \frac{R}{r_C
  + R} v_C$.
We model the $n$-input buck DC-DC converter with the DTLHS ${\cal B}_n$ = ($X$,
$U$, $Y$, $N$), with $X = [i_L$, $v_O]$, $U = [u_1$, $\ldots$, $u_n]$, $Y =
[v_D$, $v_1^D$, $\ldots, v_{n - 1}^{D}$, $i_D$, 
$q_0$, $\ldots$, $q_{n - 1}$, $I_1^{u}$, $\ldots$, $I_n^{u}$,
$ v_1^u$, $\ldots$, $v_n^u]$. 
%The transition relation $N$ is as follows.
From a simple circuit analysis 
%(e.g. see \cite{type2-fuzzy-dc-dc-2005}) 
we have the following equations:

%\vspace{-3mm}
%
\[
\begin{array}{rcl}
  \dot{i}_L & = & a_{1,1}i_L + a_{1,2} v_{O} + a_{1,3}v_D \\[\medskipamount]
  \dot{v}_{O} & = & a_{2,1}i_L + a_{2,2}v_{O} + a_{2,3}v_D
\\[\smallskipamount] 
\end{array}
%\vspace{-2mm}
\]
where the coefficients $a_{i, j}$ depend on the circuit parameters $R$, $r_L$, $r_C$, $L$ and $C$ in the following way:
$a_{1,1} = -\frac{r_L}{L}$,
$a_{1,2} = -\frac{1}{L}$,
$a_{1,3} = -\frac{1}{L}$,
$a_{2,1} = \frac{R}{r_c + R}[-\frac{r_c r_L}{L} + \frac{1}{C}]$,
$a_{2,2} = \frac{-1}{r_c + R}[\frac{r_c R}{L} + \frac{1}{C}]$,
$a_{2,3} = -\frac{1}{L}\frac{r_c R}{r_c + R}$.
Using a discrete time model with sampling time $T$ (writing $x'$ for
$x(t+1)$) we have:
%\vspace{-1mm}
\begin{eqnarray*}
  {i'_L} & = & (1 + Ta_{1,1})i_L + Ta_{1,2}v_O + Ta_{1,3}v_D  \label{buck:next-il} \\
  {v'_O} & = & Ta_{2,1}i_L + (1 + Ta_{2,2})v_O + Ta_{2,3}v_D. \label{buck:next-vc}
  %\vspace{-2mm}
\end{eqnarray*}
The algebraic constraints stemming from the constitutive equations of
the switching elements
%(e.g. see \cite{graph-theory-electrical-networks-1985})
are the following:
%\vspace{-3mm}

{\small
\[
\begin{array}{ccc}
q_0  \rightarrow  (v_D = R_{\rm on} i_D )
\!&\!
\bar{q}_0  \rightarrow   (v_D = R_{\rm off}i_D)
\!&\!
v_D   =  v_n^u - V_n
\\[\medskipamount]

q_0   \rightarrow  (i_D \geq 0)
\!&\!
\bar{q}_0  \rightarrow  (v_D \leq 0)
\!&\!
i_L =  i_D + \sum_{i=1}^{n}I_i^{u}
\\[\medskipamount]

\multicolumn{3}{c}{
\begin{array}{lr}
\bigwedge_{i \in [n - 1]} q_i  \rightarrow  (v_i^D = R_{\rm on} I_i^u)
~&~
\bigwedge_{i \in [n - 1]} \bar{q}_i  \rightarrow   (v_i^D = R_{\rm off}I_i^u)
\end{array}
}
\\[\medskipamount]

\multicolumn{3}{c}{
\begin{array}{lr}
\bigwedge_{i \in [n - 1]} q_i  \rightarrow  (I_i^u \geq 0)
~&~
\bigwedge_{i \in [n - 1]} \bar{q}_i  \rightarrow  (v_i^D \leq 0)
\end{array}
}
\\[\medskipamount]

\multicolumn{3}{c}{
\begin{array}{lr}
\bigwedge_{j \in [n]} u_j \rightarrow (v_j^u = R_{\rm on} I_j^u)
~&~
\bigwedge_{j \in [n]} \bar{u}_j  \rightarrow  (v_j^u = R_{\rm off} I_j^u)
\end{array}
}
\\[\medskipamount]

\multicolumn{3}{c}{
\bigwedge_{i \in [n]} v_D   =  v_i^u + v_{i}^{D} - V_i 
}
%\\[\smallskipamount]
%&
%
%\\[\medskipamount]
\end{array}
%\vspace{-3mm}
\]
}

\subsection{Experimental Settings}
\label{expres.tex.setting}
%We present experimental results obtained by implementing Alg.~\ref{alg:smallctr4} and 
%by running it on the inverted pendulum and multi-input buck DC-DC converter 
%described in Section~\ref{sec:inverted-pendulum-model} 
%and Section~\ref{sec:multi-input-buck-model}, respectively. 

All experiments have been carried out on an Intel(R) Xeon(R) CPU @ 2.27GHz, 
with 23GiB of RAM, Kernel: Linux 2.6.32-5-686-bigmem, distribution Debian
GNU/Linux 6.0.3 (squeeze).

As in~\cite{KB94}, we
set pendulum parameters $l$ and $m$ in such a way that  $\frac{g}{l}=1$ (i.e. $l=g$) and
$\frac{1}{ml^2}=1$ (i.e. $m=\frac{1}{l^2}$). As for the quantization, we set $A_{x_1} =
[-1.1\pi, 1.1\pi]$ and $A_{x_2} = [-4, 4]$, and we define $A_{{\cal I}_F} = A_{x_1} \times
A_{x_2} \times A_u$. The goal region is defined by the predicate 
$G_{{\cal I}_F}(X)\equiv (-\rho\leq x_1\leq\rho)\,\land\,(-\rho\leq x_2\leq\rho)$, where $\rho\in\{0.05, 0.1\}$, and the initial 
region is defined by the predicate $I_{{\cal I}_F}(X)\equiv(-\pi\leq x_1\leq \pi)\,\land\,(-4\leq x_2\leq 4$). 
%The initial region $I$ and goal region $G$ are as
%in Section.~\ref{sec:inverted-pendulum-model}, thus the DTLHS control problem we consider
%is  $P$ = (${\cal H}$, $I$, $G$).% [(${\cal L}_{\cal H}$, $I$, $G$)].

In the multi-input buck DC-DC converter with $n$ inputs ${\cal B}_n$, we set constant parameters as follows:
$L = 2 \cdot 10^{-4}$ H, $r_L = 0.1$ ${\rm \Omega}$, $r_C = 0.1$ ${\rm
\Omega}$,  $R = 5$ ${\rm \Omega}$,  $C = 5 \cdot 10^{-5}$ F,   $R_{\rm on} = 0$
${\rm \Omega}$, $R_{\rm off} = 10^4$ ${\rm \Omega}$, and $V_i = 10i$ V for $i\in
[n]$. As for the quantization, we set $A_{i_L} = [-4, 4]$ and $A_{v_O}
= [-1, 7]$, and we define $A_{{\cal B}_n} = A_{i_L} \times A_{v_O} \times A_{u_1} \times
\ldots \times A_{u_n}$. The goal
region is defined by the predicate $G_{{\cal B}_n}(X)\equiv (-2\leq
i_L\leq 2)\,\land\,(5-\rho\leq v_O\leq5+\rho)$, where $\rho = 0.01$, and
the initial  region is defined by the predicate $I_{{\cal B}_n}(X)\equiv(-2\leq i_L\leq
2)\,\land\,(0\leq v_O\leq 6.5$). 

%Moreover, 
In both examples, we use uniform quantization 
functions dividing the domain of each state variable $x$ %($x_{1}, x_2$) 
into $2^b$ equal intervals, where $b$ is the number of
bits used by AD conversion. The resulting quantizations are ${\cal Q}_{{\cal I}_F,b} = (A_{{\cal I}_F},\Gamma_b)$ and ${\cal Q}_{{\cal B}_n,b} = (A_{{\cal B}_n},\Gamma_b)$. %, with $\|\Gamma_b\| =  \frac{8}{2^b}$. 
Since %we have two quantized
in both examples have two quantized variables, 
each one with $b$ bits,  the number of quantized (abstract) states %in the control abstraction 
is exactly $2^{2b}$.   
%$=$  $(|i_L| \leq 2)$ $\wedge$ $(0 \leq v_O \leq 7)$,
%is the whole
%admissible region.  
%Thus $I$ = $(-5 \leq i_L \leq 5) \wedge (0 \leq v_O \leq
%7)$. 
%Our objective is to drive output $v_O$ to $V_{ref} = 5$ V with a tolerance
%$\theta = 0.01$ V, while keeping $|i_L| \leq 2$. Thus our goal region is
%defined by the convex predicate $G$ = $\left(4.99 \leq v_O \leq 5.01 \land -2 \leq i_L
%\leq 2\right)$. 
%DTLHS ${\cal H}$ has
%2 continuous state variables ($i_L$, $v_O$), 
%one input variable ($u$),
%4 continuous auxiliary variables ($i_u$, $i_D$, $v_u$, $v_D$) 
%and 12 boolean auxiliary variables (including $q$ as in Ex. \ref{example-buck.tex})
%stemming from the modelling of robustness with respect to $R$.
%
%Note that no (formally proved) robust control software 
%is available for buck DC-DC converters. 
%
%We note that robust (analog) control schemas for buck dc-dc converters
%have  been investigated, for example, in \cite{robust-buck-pes95,robust-buck-2008}. 
%
%
%\input{params.fig.tex}
%
%\vspace*{-1mm}
%
%\subsection{Quantization}
%
%\subsection{Results from Experiments}
%\label{results:subsec}

%\vspace*{-1mm}
We run \qks\ and \qkssc\ on the inverted pendulum model ${\cal I}_F$ 
for different values of $F$ (force intensity), 
and on the multi-input buck DC-DC model ${\cal B}_n$, 
for different values of parameter $n$ (number of the switches). 
For the inverted pendulum, we use sampling time $T=0.1$ seconds when the quantization schema has 
less than $10$ bits and $T=0.01$ seconds otherwise.
For the multi-input buck, we set $T=10^{-6}$ seconds.
For both systems, we run experiments with different quantization schema. 

For all of these experiments, \qks\ and \qkssc\ output a control software 
in C language. In the following, we will denote with $K^{\sf mgo}$ the output of \qks, 
and with $K^{\sf sc}$ the output of \qkssc\ on the same control problem.

%Experiments on the multi-input buck DC-DC converter (Tab.~\ref{tab:mibdcdc}) are parametrized 
%with respect to the number $n$ of the switches and the number of bits $b$ of the quantization schema.
%Experiments on the inverted pendulum (Tab.~\ref{tab:pendulum}) are parametrized 
%with respect to the intensity of the force $F$, the sampling time $T$, and again 
%the number of bits $b$ of the quantization schema.

\subsection{Experiments Discussion}
\label{sec:exp-disc}
We compare the controller $K^{\sf mgo}$ and $K^{\sf sc}$ by evaluating their 
size, as well as other non-functional requirements such as the set-up time and the  
ripple of the closed loop system. Tables~\ref{tab:mibdcdc} and~\ref{tab:pendulum} 
summarize our experimental results.  

\begin{figure*}
  \centering
  \hspace*{-2mm}
  \begin{tabular}{ccc}
\begin{minipage}{0.332\textwidth}
  \includegraphics[width=\columnwidth]{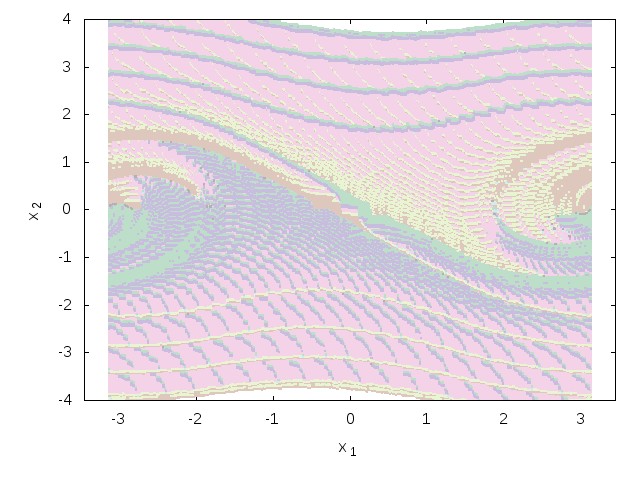}
  %\vspace{-10mm}
  \caption{$K^{\sf mgo}$ enabled actions $\!({\cal I}_{0.5},\! b\!=\!\!9$)}
  \label{fig:cr-mgo-10-05}
\end{minipage}
&
\hspace{-1.75mm}
  \begin{minipage}{0.332\textwidth}
  \includegraphics[width=\columnwidth]{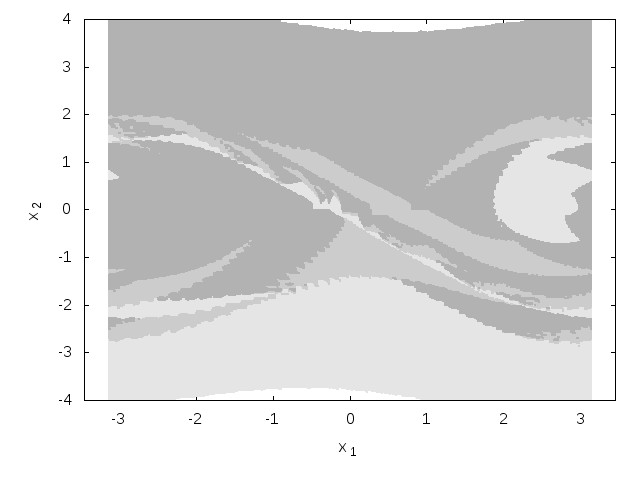}
  %\vspace{-10mm}
  \caption{$K^{\sf sc}$ enabled actions $\!({\cal I}_{0.5}, \!b\!=\!9$)}
  \label{fig:cr-sc-10-05}
  \end{minipage}
&
\hspace{-5mm}
\begin{minipage}{0.332\textwidth}
%\vspace{-2mm}
  % \includegraphics[width=\columnwidth]{pictures/simulation_setup.jpg}
  \includegraphics[width=\columnwidth]{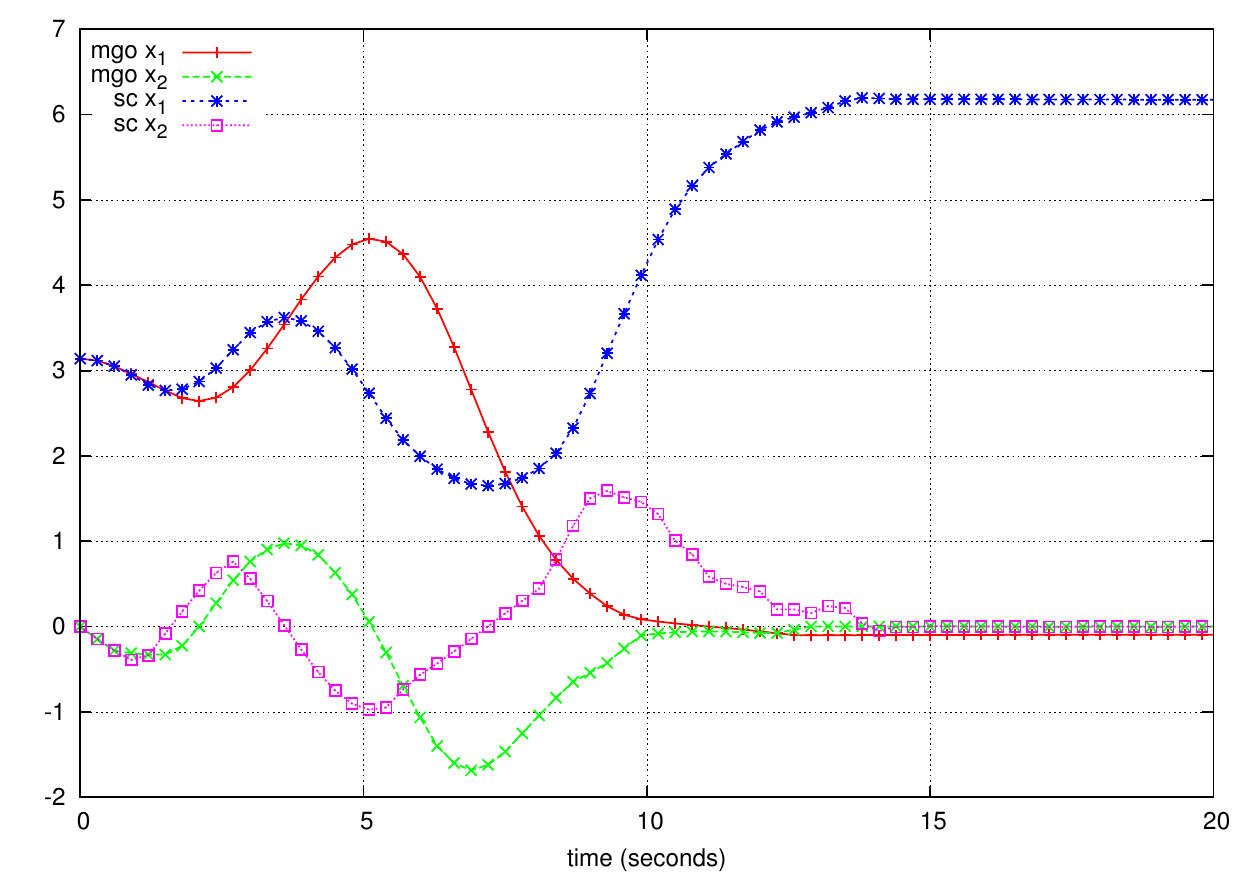}
  %\vspace{-8mm}
  \caption{Simulation of ${\cal I}_{0.5}^{K^{\sf sc}}$, ${\cal
  I}_{0.5}^{K^{\sf mgo}}\! (b\!=\!9$)}
  \label{fig:simul-setup}
\end{minipage}
\end{tabular}
%\vspace*{-0.38cm}
\end{figure*}

In both tables, column $|K^{\sf mgo}|$ (resp. $|K^{\sf sc}|$) 
shows the size (in Kbytes) of the {\small {\tt .o}} file 
obtained by compiling the output of \qks\ (resp. \qkssc) 
with {\small {\tt gcc}}. Column $\frac{|K^{\sf sc}|}{|K^{\sf mgo}|}$ shows the ratio 
between the size of the two controllers and it illustrates how much one gains in terms of 
code size by using function \fun{smallCtr} instead of \fun{mgoCtr}. 
%the usual mgo controller synthesis algorithm.

Column Path$^{\sf mgo}$ (resp. Path$^{\sf sc}$) shows the average length of 
(worst case) paths to the goal region in the closed loop abstract systems 
$\hat{\cal H}^{(K^{\sf mgo})}$ (resp. $\hat{\cal H}^{(K^{\sf sc})}$). This number, 
multiplied by the sampling time, provides a pessimistic estimation of the average set-up 
time of the closed loop system. 
Column $\frac{{\rm Path}^{\sf sc}}{{\rm Path}^{\sf mgo}}$ shows the ratio between 
the values in the two previous columns, and it provides an estimation 
of the price one has to pay (in terms of optimality) by using a small controller 
instead of the mgo controller.
 
The last three columns show the computation time of function \fun{smallCtr} (column 
CPU$^{\sf sc}$, in seconds), the ratio with respect to \fun{mgoCtr} (column 
$\frac{{\rm CPU}^{\sf sc}}{{\rm CPU}^{\sf mgo}}$), and \fun{smallCtr} memory usage (column Mem, in Kbytes). 
The function \fun{smallCtr} is obviously slower than \fun{mgoCtr}, because of 
non-optimality: it performs more loops, and it deals with more complex
computations. %data structures. 
Keep in mind, however, that the controller synthesis off-line computation is not a critical parameter in the control software synthesis flow. 

As we can see in Tab.~\ref{tab:mibdcdc} and Tab.~\ref{tab:pendulum} the size of the controller 
$K^{\sf sc}$ tends to become smaller and smaller with respect to the size of 
the correspondent controller $K^{\sf mgo}$ as the complexity of the plant model grows. 
This is a general trend, both with respect to the number of switches of the multi-input buck,  
and with respect to the number of bits of the quantization schema (in both examples). 
In particular, in the 12 bits controllers for the inverted pendulum, the size of $K^{\sf sc}$ 
is just about $5\%$ of the size of $K^{\sf mgo}$.

%A less expected result is that 
The average worst case length of paths to the goal in the closed loop system 
$\hat{\cal H}^{(K^{\sf sc})}$ 
tends to approach the one in $\hat{\cal H}^{(K^{\sf mgo})}$ as the complexity of 
the system grows. $\hat{\cal H}^{(K^{\sf sc})}$ simulations show an even better 
behaviour since most of the time, the set--up time of $\hat{\cal H}^{(K^{\sf sc})}$  
is about the one of $\hat{\cal H}^{(K^{\sf mgo})}$. 

For example, 
Fig.~\ref{fig:simul-setup} shows a simulation %(each time step being $10^{-6}$ seconds) 
of the closed loop systems ${\cal I}_{0.5}^{K^{\sf sc}}$ and ${\cal
I}_{0.5}^{K^{\sf mgo}}$. It considers a quantization schema of $9$ bits with trajectories 
starting from $x_1 = \pi, x_2 = 0$. In order to show pendulum phases, $x_1$ is
not normalized in $[-\pi, \pi]$, thus also $x_1 = 2\pi$ is in the goal. As we can see,
the small controller needs slightly more time (just about a second) to reach
the goal. This behaviour can be explained by observing that the average worst case path length is 
a very pessimistic measure. Thus, in practice, both controllers stabilize the system much faster 
than one can expect by looking at Path$^{\sf mgo}$ and Path$^{\sf sc}$.
Similarly, the performarce of 
the small controller with respect to the optimal one is much better than one can expect by considering the ratio 
$\frac{{\rm Path}^{\sf sc}}{{\rm Path}^{\sf mgo}}$.   
Interestingly, however, ${\cal I}_{0.5}^{K^{\sf mgo}}$ follows a smarter trajectory, with one less swing. 
  
Fig.~\ref{fig:ripple-mgo} (resp.~Fig.~\ref{fig:ripple-sc}) shows the ripple of $x_1$  
in the inverted pendulum closed loop system ${\cal I}_{0.5}^{K^{\sf mgo}}$ 
(resp. ${\cal I}_{0.5}^{K^{\sf sc}}$), by focusing on the part of the
simulation in Fig.~\ref{fig:simul-setup} which is (almost always) inside the goal. 
As we can see, the small controller yields a worst ripple (0.0002 vs 0.0001),
which may be however neglected in practice. %al cases.

To visualize the very different nature of these controllers, Fig.~\ref{fig:cr-mgo-10-05}  
(resp.~Fig.~\ref{fig:cr-sc-10-05}) shows actions that are enabled by $K^{\sf mgo}$ 
(resp. $K^{\sf sc}$) in all states of the admissible region of
the inverted pendulum control problem ${\cal I}_{0.5}$, by considering a quantization schema of $9$ bits. 
In these pictures, different colors mean different actions. We observe that in 
Fig.~\ref{fig:cr-mgo-10-05} we need 7 colors, because in a given state 
$K^{\sf mgo}$ may enable any nonempty subset of the set of actions.
As expected, the control strategy of $K^{\sf sc}$ is much more regular and thus simpler than the one of $K^{\sf mgo}$, since it enables the same action in 
relatively large regions of the state space. Some symmetries of Fig.~\ref{fig:cr-mgo-10-05} 
are broken in Fig.~\ref{fig:cr-sc-10-05} because when more actions could be choosen, 
\fun{smallCtr} gives always priority to one of them 
(Alg.~\ref{alg:smallctr4}, lines~\ref{line:for-start}--\ref{line:for-end}).

\begin{figure}
  \centering
  \begin{tabular}{ccc}
\begin{minipage}{0.42\textwidth}
  \includegraphics[width=\columnwidth]{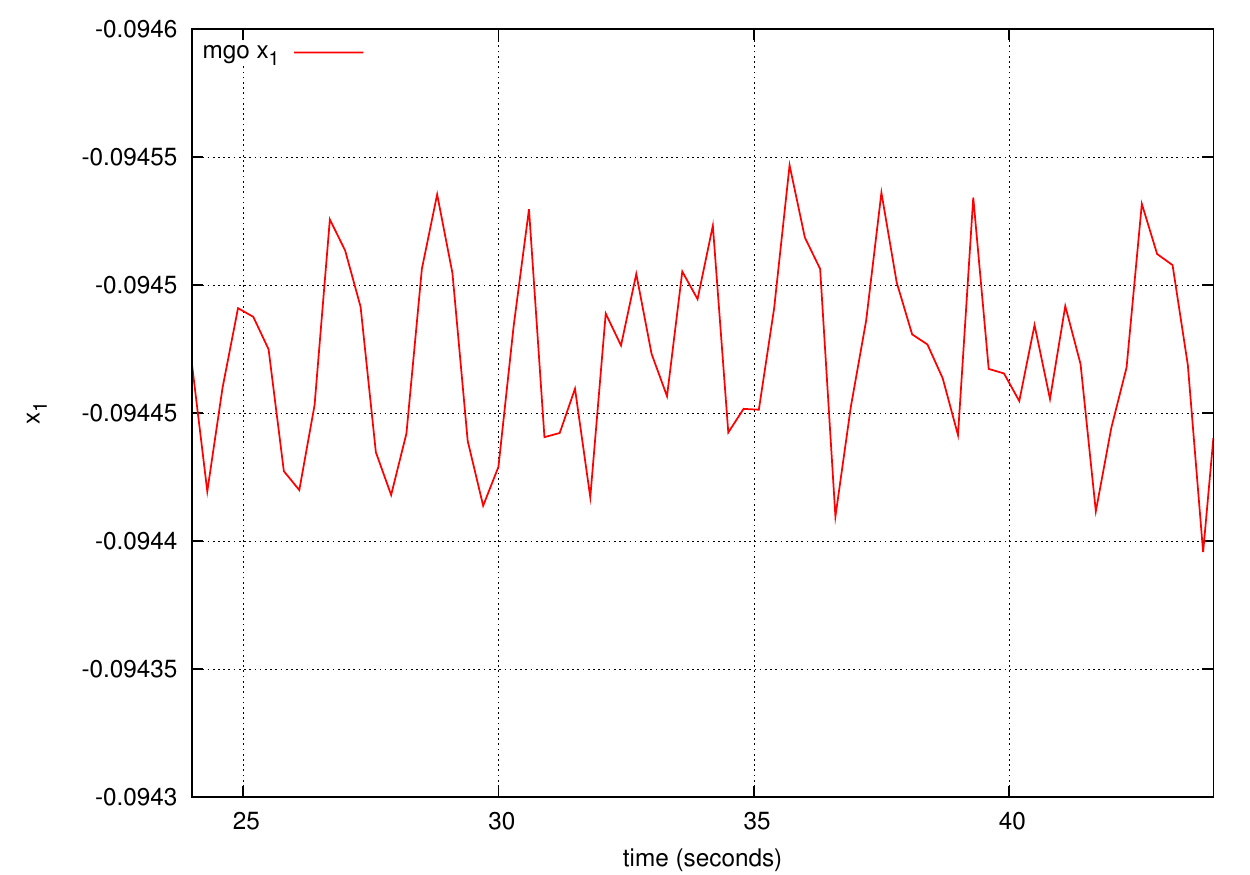}
   %\vspace{-0.6cm}	
  \caption{Ripple for $K^{\sf mgo}$ $(b\! =\! 9$)}
  \label{fig:ripple-mgo}
\end{minipage}
&
  \begin{minipage}{0.42\textwidth}
  \centering
  \includegraphics[width=\columnwidth]{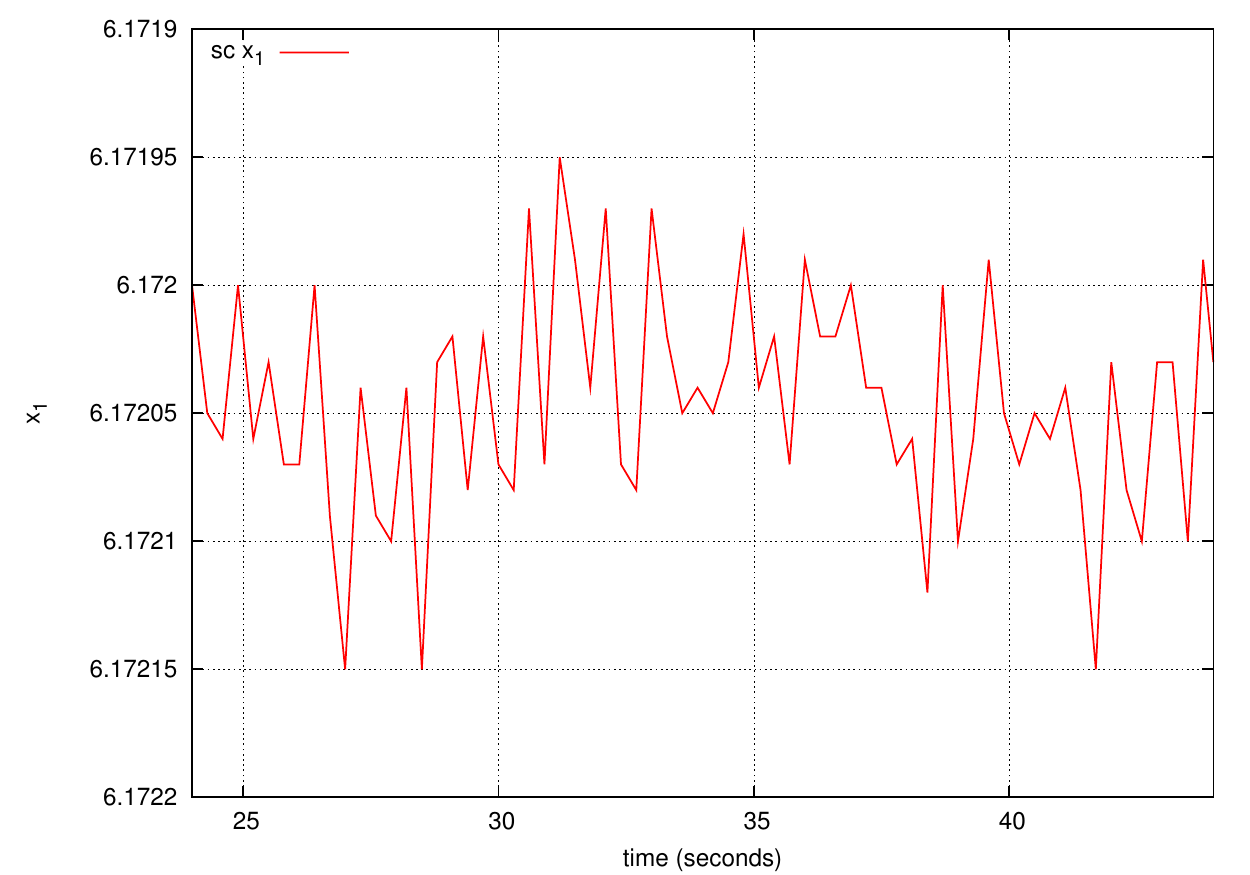}
   %\vspace{-0.6cm}	
  \caption{Ripple for $K^{\sf sc}$ ($b = 9$)}
  \label{fig:ripple-sc}
  \end{minipage}
\end{tabular}
%\vspace*{-5mm}
\end{figure}

        %\input{buck-dc-dc.tex}
        %\subsection{Multiinput Buck DC-DC Converter Model}
%\label{multi-buck-dc-dc.tex}% \Blue{Igor}}

\begin{sidewaystable}
\center
\caption{Results for Multiinput Buck DC-DC Converter}
\vspace{1mm}
\label{tab:mibdcdc}
\begin{tabular}{|*{2}{c|}|*{9}{r|}}
\hline
$b$ &$n$& $|K^{\sf mgo}|$&  $|K^{\sf sc}|$& $\frac{|K^{\sf
sc}|}{|K^{\sf mgo}|}$&  Path$^{\sf mgo}$& Path$^{\sf sc}$&
$\frac{{\rm Path}^{\sf sc}}{{\rm
Path}^{\sf mgo}}$ &CPU$^{\sf sc}$& $\frac{{\rm CPU}^{\sf sc}}{{\rm CPU}^{\rm
mgo}}$&  Mem  \\\hline\hline
9 & 1 & 36 & 30 & 83.9\% & 179.40 & 517.67 & 2.89 & 11.01 & 2.64 & 3.95e+04\\\hline
9 & 2 & 62 & 34 & 56.0\% & 142.19 & 386.70 & 2.72 & 9.15 & 1.59 & 3.71e+04\\\hline
9 & 3 & 110 & 41 & 37.3\% & 131.55 & 353.77 & 2.69 & 15.01 & 1.58 & 5.66e+04\\\hline
9 & 4 & 157 & 42 & 27.3\% & 127.53 & 324.24 & 2.54 & 19.98 & 1.37 & 6.57e+04\\\hline
10 & 1 & 91 & 56 & 61.4\% & 136.85 & 262.83 & 1.92 & 20.43 & 1.62 & 6.41e+04\\\hline
10 & 2 & 149 & 61 & 41.0\% & 110.78 & 231.37 & 2.09 & 23.14 & 1.36 & 6.71e+04\\\hline
10 & 3 & 244 & 65 & 26.9\% & 103.40 & 216.11 & 2.09 & 34.06 & 1.21 & 9.17e+04\\\hline
10 & 4 & 341 & 70 & 20.6\% & 100.43 & 209.47 & 2.09 & 53.70 & 1.18 & 1.23e+05\\\hline
\end{tabular}
%\vspace{-4mm}
\end{sidewaystable}

\begin{sidewaystable}
\center
\caption{Results for the Inverted Pendulum}
\vspace{1mm}
\label{tab:pendulum}
\begin{tabular}{|*{3}{c|}|*{9}{r|}}
\hline
$b$& $F$& $T$& $|K^{\sf mgo}|$&  $|K^{\sf sc}|$& $\frac{|K^{\sf sc}|}{|K^{\sf mgo}|}$&  Path$^{\sf mgo}$& Path$^{\sf sc}$&
$\frac{{\rm Path}^{\sf sc}}{{\rm
Path}^{\sf mgo}}$ &CPU$^{\sf sc}$& $\frac{{\rm CPU}^{\sf sc}}{{\rm CPU}^{\sf 
mgo}}$&  Mem  \\\hline\hline
8 & 0.5 & 0.1 & 163 & 44 & 27.4\% & 132.96 & 234.35 & 1.76 & 16.25 & 2.16 & 4.15e+04\\\hline
9 & 0.5 & 0.1 & 352 & 92 & 26.3\% & 69.64 & 147.74 & 2.12 & 33.59 & 2.12 & 8.47e+04\\\hline
10 & 0.5 & 0.1 & 752 & 206 & 27.5\% & 59.16 & 133.70 & 2.26 & 123.94 & 2.57 & 2.27e+05\\\hline
11 & 0.5 & 0.01 & 2467 & 213 & 8.6\% & 1315.69 & 1898.50 & 1.44 & 798.03 & 2.38 & 1.40e+05\\\hline
12 & 0.5 & 0.01 & 8329 & 439 & 5.3\% & 674.39 & 1280.32 & 1.90 & 2769.08 & 1.07 & 8.82e+05\\\hline
8 & 2.0 & 0.1 & 96 & 31 & 32.8\% & 24.30 & 58.00 & 2.39 & 3.41 & 1.87 & 4.13e+04\\\hline
9 & 2.0 & 0.1 & 185 & 81 & 44.1\% & 22.29 & 40.13 & 1.80 & 9.64 & 1.94 & 8.39e+04\\\hline
10 & 2.0 & 0.1 & 383 & 194 & 50.6\% & 21.91 & 43.24 & 1.97 & 49.26 & 2.13 & 2.25e+05\\\hline
11 & 2.0 & 0.01 & 2204 & 128 & 5.8\% & 230.25 & 437.18 & 1.90 & 198.95 & 2.87 & 1.46e+05\\\hline
12 & 2.0 & 0.01 & 5892 & 300 & 5.1\% & 207.31 & 390.48 & 1.88 & 561.18 & 0.45 & 9.63e+05\\\hline
\end{tabular}
%\vspace{-1mm}
\end{sidewaystable}

\section{Conclusions}

We presented a novel automatic methodology to synthesize  
control software for Discrete Time Linear Hybrid Systems, 
aimed at generating small size control software. 
We proved our methodology to be very effective by showing that we 
synthesize controllers up to 20 times smaller than time optimal ones.
Small controllers keep other software non-functional requirements, such as WCET, 
at the cost of being suboptimal with respect to system level non-functional requirements 
(i.e. set-up time and ripple). 
%In our experiments, they show a set-up time usually about twice longer than 
%the optimal ones. 
Such inefficiency may be fully justified since it allows a designer to 
consider much cheaper microcontroller devices.

Future work may consist of further exploiting small controller  
regularities in order to improve on other 
software as well as system level non-functional 
requirements. A more ambitious goal may consist of designing a tool 
that automatically tries to find control software that meets non-functional 
requirements given as input (such as memory, ripple, set-up time).

%An example (suggested by Remark \ref{remark:action-switches}) 
%can be software generation implementing the control strategy minimizing 
%the number of action switches.     

%\vspace*{-1mm}
%
\subsection*{Acknowledgments}
%\vspace*{-1mm}
This paper has been published in EMSOFT 2012 proceedings. We thank our anonymous referees for their helpful comments.
This work has been partially supported by the  
MIUR project TRAMP (DM24283) and by the EC FP7 projects ULISSE (GA218815) and SmartHG
 (317761).
%\vspace{-1mm}

%  \vspace{-.15cm}
  
  \bibliographystyle{plain}
  \bibliography{modelchecking}

\end{document}